%% file: main.tex
\documentclass[10pt]{article}
\pdfoutput=1
\usepackage{authblk}
\usepackage{blindtext}

\usepackage[utf8]{inputenc} % allow utf-8 input
\usepackage[T1]{fontenc}    % use 8-bit T1 fonts
\usepackage{hyperref}       % hyperlinks
\usepackage{url}            % simple URL typesetting
\usepackage{booktabs}       % professional-quality tables
\usepackage{amsfonts}       % blackboard math symbols
\usepackage{nicefrac}       % compact symbols for 1/2, etc.

\usepackage{multirow}
\usepackage{amstext}
\usepackage{amsmath}
\usepackage{amsthm}
\usepackage{amssymb}
\usepackage{bm}
\usepackage{wrapfig}
\usepackage{amsfonts, amsmath}
\usepackage[utf8]{inputenc}
\usepackage{graphicx}
\usepackage{bbm, comment}
\usepackage{color}
\usepackage{float}
\usepackage{wrapfig}
\usepackage{url}
\usepackage{MnSymbol,wasysym,marvosym,ifsym}
\usepackage[T1]{fontenc}
\usepackage[inline]{enumitem}

\usepackage{microtype}      % microtypography
\usepackage{booktabs} % For formal tables
\usepackage[ruled]{algorithm2e} % For algorithms
\usepackage[tight]{subfigure}
\usepackage{tikz}
\usepackage{threeparttable}
\usepackage[utf8]{inputenc}
\usepackage{graphicx} 
\usepackage{float} 
\usepackage{pgfplots}
\pgfplotsset{compat=1.14}
\usepackage{caption}
\usepackage{graphicx}

\makeatletter
\def\BState{\State\hskip-\ALG@thistlm}
\makeatother

\usepackage{scalerel,stackengine}
\stackMath
\newcommand\reallywidehat[1]{%
\savestack{\tmpbox}{\stretchto{%
  \scaleto{%
    \scalerel*[\widthof{\ensuremath{#1}}]{\kern-.6pt\bigwedge\kern-.6pt}%
    {\rule[-\textheight/2]{1ex}{\textheight}}%WIDTH-LIMITED BIG WEDGE
  }{\textheight}% 
}{0.5ex}}%
\stackon[1pt]{#1}{\tmpbox}%
}
\usepackage{geometry}
\usepackage[numbers]{natbib}

\newtheorem*{theorem*}{Theorem}
\newtheorem{theorem}{Theorem}[section]
\newtheorem{lemma}[theorem]{Lemma}

\newtheorem{example}[theorem]{Example}
\newtheorem{proposition}[theorem]{Proposition}
\newtheorem{definition}[theorem]{Definition}

\usetikzlibrary{external}
\usepgfplotslibrary{fillbetween}
\newcommand*\Defs{
    set samples 1001;
    min(p,q)=(p<q)?p:q;
    max(p,q)=(p>q)?p:q;
    const(x,l,r,w)=x<=l||x>=r?0:w;
    Binv(p,q)=exp(lgamma(p+q)-lgamma(p)-lgamma(q));
    beta(x,p,q)=p<=0||q<=0?1/0:x<0||x>1?0.0:Binv(p,q)*x**(p-1.0)*(1.0-x)**(q-1.0);
}

\SetAlFnt{\small}
\SetAlCapFnt{\small}
\SetAlCapNameFnt{\small}
\SetAlCapHSkip{0pt}
\IncMargin{-\parindent}

\usepackage{xcolor}

\begin{document}

\title{Equal Affection or Random Selection: the Quality of Subjective Feedback from a Group Perspective}

%\author[ ]{Jiale Chen}
%\author[ ]{Yuqing Kong}
%\author[ ]{Yuxuan Lu}
%\author{Jiale Chen\and Yuqing Kong\and Yuxuan Lu}
\author[ ]{Jiale Chen \qquad Yuqing Kong \qquad Yuxuan Lu}

\affil[ ]{The Center on Frontiers of Computing Studies, Computer Science Dept., Peking University}
\affil[ ]{\texttt{\{jiale\_chen, yuqing.kong, yx\_lu\}@pku.edu.cn}}
%\eid{123@gmail.com}
\date{}
\maketitle

% Abstract. Note that this must come before \maketitle.
\begin{abstract}

In the setting where a group of agents is asked a single subjective multi-choice question (e.g. which one do you prefer? cat or dog?), we are interested in evaluating the quality of the collected feedback. However, the collected statistics are not sufficient to reflect how informative the feedback is since fully informative feedback (equal affection of the choices) and fully uninformative feedback (random selection) have the same uniform statistics. 

Here we distinguish the above two scenarios by additionally asking for respondents' predictions about others' choices. We assume that informative respondents' predictions strongly depend on their own choices while uninformative respondents' do not. With this assumption, we propose a new definition for uninformative feedback and correspondingly design a family of evaluation metrics, called $f$-variety, for group-level feedback which can 1) distinguish informative feedback and uninformative feedback (separation) even if their statistics are both uniform and 2) decrease as the ratio of uninformative respondents increases (monotonicity). We validate our approach both theoretically and numerically. Moreover, we conduct two real-world case studies about 1) comparisons about athletes and 2) comparisons about stand-up comedians to show the superiority of our approach. 
\end{abstract}

\section{introduction}

%the quality evaluation of the subjective feedback remains to be a challenge since no ground truth exists. For example, even if researchers develop payment schemes that theoretically encourage high-quality feedback, we still need a systematic way to compare the quality of the subjective feedback elicited by different payment schemes.

%Evaluating feedback quality is important but also difficult in subjective surveys since no ground truth is provided.

%For example, the political pollsters can use this metric to evaluate the informativeness of the polls from a group of voters in a specific region. This metric can be also used to evaluate other subjective surveys (e.g. survey for purchase intention~\cite{radas2019whose}) quality from a group of people.

Many areas need subjective data collected by survey methods. For example, the voting pollsters need to elicit subjective opinions from potential voters~\cite{rothschild2011forecasting,galesic2018asking}. Product companies want to elicit purchase intentions from potential customers~\cite{radas2019whose}. Meanwhile, it is always a concern that the feedback quality may not be guaranteed due to the lack of expertise or effort of the respondents~\cite{radas2019whose}. In such case, researchers develop multiple approaches (e.g. attention test, different non-flat payment schemes)~\cite{prelec2004bayesian,weaver2013creating,witkowski2012robust,radanovic2014incentives, DBLP:conf/innovations/KongS18} to encourage high-quality subjective feedback. However, despite the existence of these elicitation approaches, there does not exist a systematic way to evaluate the quality of the subjective feedback such that we can compare those elicitation approaches in practice. 

A key challenge here is that we cannot verify each individual's answer since it is subjective. From a group perspective, the collected statistics are not sufficient to reflect how informative they are. For example, given a multi-choice question (e.g. Which one do you prefer? Panda Express or Chick-fil-A), an unbalanced statistics (e.g. 80\% Chick-fil-A, 20\% Panda Express) is informative since uninformative respondents' statistics should be uniform (50\%,50\% for binary-choice). However, the opposite is not true. That is, uniform statistics may not be uninformative. Let's consider the following example. 

\begin{itemize}
\item Which one do you prefer? dog or cat?
\item Which one do you prefer? realism or liberalism?
\end{itemize}

When we ask people the above two questions, we may receive uniform statistics (50\%,50\%) for both of them. However, it's possible that in the first question, the respondents fully understand the question and half of them prefer cats while in the second question, the respondents do not understand the question and randomly select one choice. To distinguish ``equal affection'' and ``random selection'', we need 1) to collect additional information besides their choices; 2) a more refined concept of \emph{uninformativeness}. 

To address the above problem, we additionally probe the respondents' predictions for other people by asking them "What percentage of people prefer dogs?". With this additional query, we can compare the dog lovers' predictions and cat lovers' predictions (also for realism/liberalism ``lovers''). Since people are usually attracted to belief systems that are consistent with their preferences, we can boldly assume that dog lovers have very different predictions for other people's preferences from cat lovers. However, if respondents do not understand the meaning of realism nor liberalism, they will not form strong opinions for other people's preferences. 

Inspired by this, we utilize the additional statistics about people's predictions and propose a more refined concept of uninformativeness by adding a condition that describes the relationship between respondent's choice and prediction. In our new definition, a respondent's feedback is uninformative if and only if 
\begin{itemize}
\item \emph{Uniform choice}: she picks the choice uniformly at random;
\item \emph{Independence}: her choice and prediction are independent. 
\end{itemize}

With the above definition, the ``random selection'' is still uninformative while the ``equal affection'' is not. Moreover, we show that this new definition satisfies two natural properties:
\begin{itemize}
\item \emph{Stability}: a mixed group of uninformative feedback is still uninformative;
\item \emph{Additive property}: a mixed group of uninformative feedback and informative feedback is informative. 
\end{itemize}

We also provide a corresponding family of \emph{non-negative} evaluation metric, \emph{$f$-variety}, such that with our new definition, $f$-variety
\begin{itemize}
\item \emph{Separation:} separates informative and uninformative feedback by assigning approximately zero value to only uninformative feedback;
\item \emph{Monotonicity:} decreases as the ratio of uninformative feedback increases. 
\end{itemize}

$f$-variety is defined as $f$-divergence between the joint distribution over choice-prediction pairs and the corresponding uninformative ones, which has uniform choices and the same marginal distribution over predictions. Intuitively, $f$-variety represents the amount of information contained in group-level feedback. To give a taste of $f$-variety, we take a special $f$-divergence, total variation distance (Tvd), as an example and visualize the corresponding Tvd-variety in the following figure.

\input{fig/tvd}

In addition to theoretical validation (Section~\ref{sec:theory}), we perform multiple numerical experiments (Section~\ref{sec:num}) to validate the robustness of $f$-variety when we only have access to samples rather than the joint distribution over choice-prediction pairs. We also perform two real-world case studies about comparisons for athletes and comparisons for stand-up comedians (Section~\ref{sec:cases}). For evaluation, we also collect side information of our respondents as reference (e.g. we ask for their knowledge about these two contents in advance). We compare $f$-variety with a baseline, which only measures the uniformity of the aggregated statistics of the choices. The results show that compared to baseline, $f$-variety is more consistent with the reference, which shows the superiority of $f$-variety. In the situation where we cannot obtain high-quality side information (e.g. polls, survey purchase intention, comparing different payment schemes for subjective surveys), we can use $f$-variety as an evaluation metric.

\subsection{Related work}

We use the choice-prediction framework for data collection, {\sl i.e.}, we ask for respondents' choices and their predictions of others' choices. The choice-prediction framework has applications in different fields. Firstly, it provides more accurate data. Psychological research suggests that peer-predictions (predictions about others' behavior) are a more accurate predictor of individuals' future behavior than self-predictions (predictions about oneself)~\cite{helzer2012and}. Researches about political voting also have shown that predictions of others' choices can achieve higher accuracy of election predictions~\cite{rothschild2011forecasting,galesic2018asking}. We consider a different problem and focus on evaluating the collected data without any ground truth. 

Secondly, the choice-prediction framework is used to elicit truthful opinions of people in information elicitation. Bayesian Truth Serum (BTS) combines respondents' choices and predictions and then creates incentives for truthfulness in elicitation for subjective questions~\cite{prelec2004bayesian, weaver2013creating}. And it's directly used in solving the crowd wisdom questions~\cite{prelec2017solution}. There are also further works that adopt the same framework while avoiding BTS's assumption of infinite participants~\cite{witkowski2012robust,radanovic2014incentives,DBLP:conf/innovations/KongS18}. These works all focus on designing truthful incentives for individuals and assume that people who have the same choices also have the same predictions, {\sl i.e.}, the common prior assumption. In contrast to the above works, we focus on designing an evaluation metric after collecting feedback from a group of people and do not need the strong common prior assumption.

Thirdly, the choice-prediction framework can also be used in measuring the expertise of individual~\cite{radas2019whose}, which is closely related to our work. \citet{radas2019whose} assume that people with more accurate predictions are more informative and use this assumption to measure the expertise of each individual. This assumption may not be valid when experts are not familiar with other people who also answer the question. Our work measures the expertise of a \emph{group} of people and we do not need such an assumption. 

%Also, we use a completely different method, $f$-variety.

Our work uses $f$-divergence as an important ingredient to designing the new metric, $f$-variety. Previous works are using $f$-divergence in measuring the amount of information in individuals' answers to subjective questions~\cite{kong2018water,kong2019information}. These works aim to elicit individuals' truthful opinions, while our work uses $f$-divergence to measure group-level informativeness and leads to a totally different metric.

\section{Theory}\label{sec:theory}

In this section, we will formally introduce our model and state our definition for uninformative distribution. Given the definition, we will propose a family of metrics, $f$-variety, to measure the amount of information contained in distributions. We will provide a theoretical validation for both our definition and our metrics.  

A group of agents is asked to answer a multi-choice question (e.g. which one do you prefer? realism or liberalism?) and also predict other people's choices (e.g. what percentage of people prefer realism?). Given the question, we assume that each agent receives a pair of choice and prediction $(c,p)$ from distribution $D$ independently. We do not assume that all agents are homogeneous. That is, Alice's choice-prediction pair can have a different distribution from Bob's. For a group of agents, we care about the \emph{average} distribution over their choice-prediction pairs. %The average distribution can be also understood as the distribution over a random agent's choice-prediction pair. 

%We use $(C,P)$ to denote a random agent's choice-prediction pair and $D$ to denote the joint distribution over $(C,P)$. 

%We assume that given the question, participants receive pairs of choice and prediction $(c,p)$ i.i.d. from distribution $D$. We use $(C,P)$ to denote a random choice-prediction pair that is drawn from $D$. 

For non-experts who have no clue about the question's meaning, they will pick the choice uniformly at random. We can define a distribution with uniform choices as an uninformative distribution. However, like our motivating example, experts' feedback can also be uniform (e.g. 50\% experts prefer realism). In this case, we refine the previous definition by additionally requiring that a non-expert's choice is independent of her prediction. Formally, we require that every non-expert's choice-prediction pair is drawn from an \emph{uninformative} distribution which is defined as follows:

\begin{definition}[Uninformative Dist $U\otimes P$]
A distribution $D$ over choice and prediction $C,P$ is uninformative if and only if:
\begin{itemize}
\item Uniform choice: the marginal distribution of choice is uniform, {\sl i.e.} $\Pr_{D}[C=c]=\frac{1}{N_C}$ where $N_C$ is the number of choices;
\item Independence: Choice and prediction are independent, {\sl i.e.} $\Pr_{D}[C=c,P=p]=\Pr_{D}[C=c]\Pr_D[P=p]$
\end{itemize}
\end{definition}

Given random variables $X$, $Y$, we use $X\otimes Y$ to represent the independent joint distribution which is the product of $X$ and $Y$'s distributions. We use $U$ to denote a random choice whose distribution is uniform. An uninformative distribution $D$ can be represented by $U\otimes P$ where $P$'s distribution is $D$'s marginal distribution over the predictions.

\input{fig/groups}
We use four different distributions in Figure~\ref{fig:groups} to explain our definition of uninformative distribution. The definition $U\otimes P$ not only is consistent with our intuition but also has multiple desired natural properties. The first property, \emph{stability}, is that a mixed group of non-experts is still uninformative. In the previous example, mixing a group of respondents, whose average distribution is dist 3, and another group, whose average distribution is dist 4, will not make them informative. The second property, \emph{additive property}, means that adding experts into the group of non-experts will make the whole group informative. The initial definition, which defines a distribution with uniform choices as an uninformative distribution, satisfies both stability and additive property naturally. We show that our refined definition still satisfies the two properties and allows a more refined concept of non-experts. %On the other hand, if we define uninformative dist as $C\otimes P$ with only the independence condition, the following natural properties will not be satisfied \yk{provide a concrete counterexample in the appendix and refer that here}. 

\begin{proposition}[Properties of $U\otimes P$]
The average distribution over 
\begin{description}
\item [Stability: 0+0=0] a mixed group of non-experts' choice-prediction pairs is uninformative;
\item [Additive Property: 0+!0=!0] a mixed group of experts and non-experts' choice-prediction pairs is informative.
\end{description}
\end{proposition}
\begin{proof}
Given two group of non-experts whose average distributions are $U\otimes P_1$ and $U\otimes P_2$ correspondingly, the mixed average distribution will be $\alpha U\otimes P_1 + (1-\alpha) U\otimes P_2=U\otimes (\alpha P_1 + (1-\alpha) P_2)$ since $\alpha \Pr[U=u]\Pr[P_1=p]+(1-\alpha)\Pr[U=u]\Pr[P_2=p]=\Pr[U=u]\left(\alpha \Pr[P_1=p]+(1-\alpha)\Pr[P_2=p]\right)$. Thus, the mixed average distribution is still uninformative. 
Given a group of non-experts $U\otimes P_0$ and a group of experts $CP$, if the average distribution over the experts has non-uniform marginal distribution over the choices, then the mixed version must have non-uniform marginal distribution over choices as well thus be informative. Therefore, we only need to consider the situation where the average distribution over the experts is $UP$. In this case, we will prove the result by contradiction. Let's assume that the mixed version has uninformative average distribution. Then there exists a random variable $P_{\text{mix}}$ and $\alpha>0$ such that $\alpha \Pr[U=u]\Pr[P_0=p]+(1-\alpha)\Pr[U=u,P=p]=\Pr[U=u]\Pr[P_{\text{mix}}=p]$ where $\Pr[P_{\text{mix}}=p]=\alpha \Pr[P_0=p]+(1-\alpha)\Pr[P=p]$. This implies that $\Pr[U=u,P=p]=\Pr[U=u]\Pr[P=p]$ which contradicts the fact that $UP$ is informative, {\sl i.e.}, not equal to $U\otimes P$. 
\end{proof}

Given the definition of uninformative distribution, it's natural to ask for a metric for the informativeness of the distribution. At a high level, this metric should be always non-negative and assign zero value to uninformative distribution and strictly positive value to informative distribution. Moreover, we want the metric to satisfy an information-monotonicity as well: mixing experts with non-experts will decrease the amount of information contained in experts.

We propose the following metric family, $f$-variety, that satisfies all desired properties. The idea is to measure the amount of information contained in a distribution $D$ by measuring its ``distance'' to a corresponding uninformative distribution. To measure the ``distance'', we use $f$-divergence $D_f:\Delta_{\Sigma}\times \Delta_{\Sigma}\rightarrow \mathbb{R}$, a non-symmetric measure of the difference between distribution $\mathbf{p}\in \Delta_{\Sigma} $ and distribution $\mathbf{q}\in \Delta_{\Sigma} $ %. $f$-divergence of $\mathbf{p}$ and $\mathbf{q}$
and is defined to be $$D_f(\mathbf{p},\mathbf{q})=\sum_{\sigma\in \Sigma}
\mathbf{p}(\sigma)f\left( \frac{\mathbf{q}(\sigma)}{\mathbf{p}(\sigma)}\right)$$
where $f(\cdot)$ is a convex function and $f(1)=0$. Two commonly used $f$-divergences are KL divergence $D_{KL}(\mathbf{p},\mathbf{q})=\sum_{\sigma}\mathbf{p}(\sigma)\log\frac{\mathbf{p}(\sigma)}{\mathbf{q}(\sigma)}$ by choosing $-\log(x)$ as the convex function $f(x)$, and Total variation Distance $D_{tvd}(\mathbf{p},\mathbf{q})=\frac{1}{2}|\mathbf{p}-\mathbf{q}|_1=\frac{1}{2}\sum_{\sigma}|\mathbf{p}(\sigma)-\mathbf{q}(\sigma)|$, by choosing $\frac{1}{2}|x-1|$ as the convex function $f(x)$.

\begin{definition}[$f$-variety]
For any distribution $D$ over choice and prediction, we define the $f$-variety of $D$ as \[ V^f(D):=D_f(CP,U\otimes P) \] where $CP$ represents distribution $D$ and $U\otimes P$ represents the uninformative distribution which has the same marginal distribution over predictions as $D$. 
\end{definition}

\paragraph{$f$-variety vs $f$-mutual information} The definition of $f$-variety is very similar to the definition of $f$-mutual information $D_f(CP,C\otimes P)$. In the concept of mutual information, the uninformative joint distribution is the distribution over two independent random variables. Thus, mutual information measures the information of a joint distribution $CP$ by measuring the distance between $CP$ and $C\otimes P$. A natural question here is that in our setting, can we extend the definition of uninformative distribution to $C\otimes P$ and use the $f$-mutual information between the choice and prediction to measure the informativeness. The answer is no since $C\otimes P$ does not satisfy the stability property $0+0=0$. For example, both dist 1 and 4 in Figure~\ref{fig:groups} have independent choice and prediction. However, a mixed version of them does not. Without satisfies the stability property, the monotonicity will never be satisfied since adding non-experts can increase informativeness. 

We introduce a special $f$-variety, Tvd-variety. This special measure has a nice visualization in the binary case (see Figure~\ref{fig:tvd}).

\begin{example}[Tvd-variety]Given $D$, we use vector $\mathbf{q}$ to represent the marginal distribution vector over predictions. We use vector $\mathbf{q}_c$ to represent the distribution vector over predictions, conditioning on the agent receives choice $c$.

\begin{align*}
V^{tvd}(D)=&D_{tvd}(CP,U\otimes P)\\
=& \frac{1}{2} \sum_{c,p} |\Pr[C=c,P=p] - \frac{1}{N_C} \Pr[P=p]|\\
= & \frac{1}{2} \sum_c |q_c \mathbf{q}_c  - \frac{1}{N_C} \mathbf{q} |_1
\end{align*} where $N_C$ is the number of choices. 
In the binary choice case, 
\begin{align*}
V^{tvd}(D)=& \frac{1}{2} \sum_{c=+-} |q_c \mathbf{q}_c  - \frac{1}{2} \mathbf{q} |_1\\
= & \frac{1}{2} \sum_{c=+-} |q_c \mathbf{q}_c  - \frac{1}{2} (q_+ \mathbf{q}_+ + q_- \mathbf{q}_-) |_1\\
= & \frac{1}{2} |q_+ \mathbf{q}_+ - q_- \mathbf{q}_-|_1
\end{align*}
Thus, in binary case, Tvd-variety is half of the area of the symmetric difference of red and blue regions (see Figure~\ref{fig:tvd}). 
\end{example}

Here we formally state and prove the properties of the general $f$-variety. 

\begin{theorem}[Properties of $f$-variety]
$f$-variety $V^f$ satisfies:
\begin{description}
\item [Separation: $V^f$(0)=0, $V^f$(!0)>0]  for any $D$, $V^f(D)\geq 0$, for any uninformative $D_0$, $V^f(D_0)=0$;
\item [Monotonicity: $V^f$(x+0)<$V^f$(x)] for any $D$ and any uninformative $D_0$, $\forall 0<\alpha<1$, \[V^f((1-\alpha) D + \alpha D_0 )\leq (1-\alpha) V^f(D)\]
\end{description}
\end{theorem}

\begin{proof}
The separation property follows directly from the definition of $f$-variety and uninformative distribution. 

To prove monotonicity, we need to use the joint convexity of $f$-divergence. 
\begin{lemma}[Joint Convexity~\cite{csiszar2004information}]\label{fact:jointconvexity}
For any $0\leq\lambda\leq 1$, for any $\mathbf{p_1},\mathbf{p_2},\mathbf{q}_1,\mathbf{q}_2\in \Delta_{\Sigma}$, $$D_f(\lambda \mathbf{p_1}+(1-\lambda)\mathbf{p_2},\lambda \mathbf{q_1}+(1-\lambda)\mathbf{q_2})\leq \lambda D_f(\mathbf{p_1},\mathbf{q_1})+(1-\lambda)D_f(\mathbf{p_2},\mathbf{q_2}).$$
\end{lemma}
With the above lemma, 
\begin{align*}
V^f((1-\alpha) D + \alpha D_0) = & D_f((1-\alpha) CP + \alpha U\otimes P_0, U\otimes ((1-\alpha)P+\alpha P_0))\\ 
= & D_f((1-\alpha) CP + \alpha U\otimes P_0, (1-\alpha) U\otimes P+\alpha U\otimes P_0))\\
\leq & (1-\alpha) V^f(D)
\end{align*}

\end{proof}

The above theorem implies that if we can estimate the average distribution of a group of agents perfectly and use it to calculate $f$-variety, $f$-variety can separate experts and non-experts and satisfy information-monotonicity perfectly. Since in this case, non-experts' $f$-variety will be zero, and experts' $f$-variety will be a positive number. Moreover, adding non-experts into an existed group will decrease the $f$-variety.

However, we cannot obtain a perfect estimation of distribution in practice since we only have a finite number of samples. In practice, when we ask for additional prediction, we provide the respondents 11 discrete options $\left\{0\%,10\%,...,100\%\right\}$ and use the empirical histogram to estimate the distribution and calculate $f$-variety. We will provide several numerical experiments to show the robustness of our empirical estimation method in Section~\ref{sec:num}.

%As a result, we use a discrete version of the divergence. Under common circumstances, when respondents are asked to estimate percentage, we provide them with 11 options, from 0\% to 100\%, one option for every 10\% difference. In this case, $\Sigma=\left\{0\%,10\%,...,100\%\right\}$.

\input{fig/discrete}

%\begin{theorem}When the number of participants goes to infinite:\begin{description}\item [Separation]: non-experts' $f$-variety goes to zero and experts' $f$-variety goes to a positive number;\item [Monotonicity] adding non-experts will decrease the $f$-variety. \end{description}\end{theorem}

%

\section{Numerical experiments}\label{sec:num}

%\yk{design numerical experiments to validate the theory, show a curve that represents how polarity changes with respect to alpha}

%\cjl{need refs}

%$f$-variety has the perfect separation and monotonicity property, given its input is the true average distribution. However, we may only get finite samples and calculate the $f$-variety of the empirical average distribution in practice. 

%This section will test the robustness of the empirical $f$-variety using numerical experiments by checking whether it satisfies the separation and monotonicity properties. 

In this section, we will generate multiple choice-prediction pairs of experts/non-experts and mix them with different ratios. Ideally, the $f$-variety will decrease as the ratio of non-experts increases (monotonicity) and vanish when there are only non-experts (separation). We will test the robustness of the empirical distribution's $f$-variety by checking whether it satisfies the monotonicity and separation property. 

To generate the synthetic data, we first determine the underlying distributions of the experts and non-experts. We then generate the choice-prediction pairs according to the underlying distribution for experts and non-experts, regarding different sample sizes (e.g. 100, 200, 500, 1000). We conduct multiple numerical experiments with different underlying distributions of the experts. 

We test the empirical Tvd-variety\footnote{We also test Pearson-variety and Hellinger-variety. The results are similar and shown in the Appendix.} by performing 4 groups of experiments and use Beta distribution~\cite{johnson1995continuous} to model the underlying conditional distribution over the predictions. For all cases, we choose Beta(2,2) as the distribution over non-experts' predictions while we use different distributions over experts' choice-prediction pairs in different groups. We show the results in Figure \ref{fig:numfig}. It shows that empirical Tvd-variety approximately decreases with the ratio of non-experts. The empirical Tvd-variety becomes closer to true Tvd-variety as sample size goes larger\footnote{The error bar in the figure shows the standard deviation.}. 

\input{numerical_fig/tvd-variety}
%Figure \ref{fig:numfig} shows that Tvd-variety decreases with the ratio of non-experts. It can be seen that when the sample size grows larger, the {\bf Separation} and {\bf Monotonicity} gradually emerge.\footnote{The error bar in the figure is calculated by the standard deviation.}
%It verifies theory 3.7 and it's safe to estimate that 500 samples are enough for the practical use of $f$-variety according to the simulation.

\section{Case studies}\label{sec:cases}

We perform two real-world case studies. In each study, we pick a topic and design an online survey about this topic. Each survey consists of multiple subjective questions in the format of “Which one do you prefer? $X$ or $Y$? What percentage of people will choose $X$?”. $(X,Y)$ pair represents comparable athletes, stand-up comedians, or other concepts. The orders of options are randomly shuffled. For the prediction question, we provide 11 prediction options $\{0\%,10\%,\cdots,100\%\}$. We conduct the survey on an online survey platform and our respondents are recruited by the platform. Each survey contains an attention test\footnote{Each attention test has the following form: ``There are $n$ red balls and $m$ blue balls with the same shape in the box. One is randomly selected. What percentage do you think is the probability of a red/blue ball?''} and rewards respondents who pass the test a flat participation fee. In total, we ask 15 questions
. Each question is answered by above 600 respondents on average and we pay \$0.5 for each answer sheet. 

%\noindent{\sl Experiment settings.}
%Each survey consists of a single questionnaire on the website which respondents were asked to complete. Respondents are informed to fill the questionnaire independently. There was no time limit for the study. The questionnaire contained several subjective questions, each consisting of “Which one do you prefer? X or Y?” for every (X,Y) pair could representing comparable celebrities, athletes, or other concepts. For example, one of the questions was "Which one do you prefer? Dog or Cat?" And another one was "Which one do you prefer? Realism or Liberalism?" The questions were arranged randomly and the orders of options were randomly shuffled. For each question, respondents gave the option X or Y representing their preferences. Additionally, respondents were asked to estimate the percentage of participants in the experiment who will choose X for each question.

%\noindent{\sl Respondents and procedure.}Respondents were people who actively fill out questionnaires on online platforms to seek remuneration. All respondents knew they would receive a flat participation fee before filling out the questionnaire.

\paragraph{Evaluation}
%For evaluation, we also collect other information of the respondents and divide them into two groups where it's commonly assumed that one group has more expertise than another one. For example, for the study about athletes, we ask them whether they often watch sports. In such a case, we can evaluate the informativeness metric by checking whether the high-expertise group has a higher metric. 

%The baseline index uses the option with the largest gap from a uniform distribution to describe how informative the group is about a question.

We choose Tvd-variety in the analysis of the studies. We will evaluate Tvd-variety through two aspects:
\begin{itemize}
    \item {\bf Cross respondents:} We divided the respondents into two groups through side questions (e.g. do you watch sports frequently). It is commonly assumed that one group is more familiar with the questions than the other. We can evaluate the informativeness metric by checking whether a high-expertise group has a higher metric. %We compare the number of questions of Tvd-variety and baseline that in line with common sense.
    \item {\bf Cross questions:} We divide the questions into easy and hard categories in advance. We can evaluate the informativeness metric by checking whether the easy questions have higher metrics than the hard ones.
\end{itemize}

We also compare Tvd-variety with a baseline metric which measures the degree of unbalance of the statistics. Our case studies focus on binary cases. In the binary case, the baseline metric is defined as \[Baseline:=|q_+-50\%|.\] 

More uniform statistics will have a lower baseline score. 

%The baseline metric measures the degree of unbalance of the statistics and is defined as \[Baseline:=\max_{c}|\Pr[C=c]-\frac{1}{N_c}|.\] Since our case studies focus on binary case, the baseline metric degenerates to $|q_+-50\%|.$

\subsection{Case study for athletes}
We conducted a study about the preference for athletes. We asked 7 questions and there were 656 respondents, of which 306 were men and 350 were women. Additionally, we asked respondents whether they often watch sports. 215 of respondents reported they often watch sports and the others reported they do not.

Here we will give a sample question. All 7 questions have the same format. We attach the contents of these questions in the appendix.
\begin{itemize}
    \item Which soccer player do you prefer? Andrés Iniesta or Luka Modric?
    \item What percentage of people do you think prefer Andrés Iniesta?
\end{itemize}

Figure \ref{ath:watch} shows the comparison between Tvd-variety and the baseline. In the performance comparison in the groups of often watching sports and of not, the baseline correctly suggests that people who often watch sports are more informative for 4 of the questions, but gives the opposite results for the other 3 questions. In contrast, among all 7 questions, Tvd-variety correctly suggests that people who often watch sports are more informative for 6 questions.
\input{case_athletes/watch}

%It can be seen that when the choice deviation is large, the baseline has a good performance, otherwise, the baseline algorithm is prone to errors. Yet, Tvd-variety can give good performance in both situations.

 %For the questions about soccer(M), snooker, the high-expertise group also has uniform preferences. Thus, the baseline cannot distinguish the two groups while our metric, Tvd-variety, can. 

\subsection{Case study for stand-up comedians}
We conducted a study about the preference for stand-up comedians. We asked 8 questions, 4 of which compares native stand-up comedians (type native), other 4 of which compares foreign stand-up comedians (type foreign). There were 632 respondents, of which 262 were men and 370 were women. Additionally, We asked respondents the frequency they watch native/foreign stand-up comedy. The result was shown in the chart below.
\begin{table}[h]
\centering
\begin{tabular}{|l|l|l|l|l|}
\hline
& often & sometimes & occasionally & almost never \\ \hline
native stand-up comedy & 240 & 252 & 123 & 17 \\ \hline
foreign stand-up comedy & 21 & 169 & 248 & 194 \\ \hline
\end{tabular}\caption{\small Frequency of watching stand-up comedy}\label{table:fre}
\end{table}

Here we will give a sample question. All 8 questions have the same format. We attach the contents of these questions in the appendix.
\begin{itemize}
    \item Which stand-up comedian do you prefer? Ronny Chieng or Jimmy OYang?
    \item What percentage of people do you think prefer Ronny Chieng?
\end{itemize}

\paragraph{Cross respondent} For each type, we divide the respondents into two groups. The familiar group for type native is defined as the group of respondents who report that they often or sometimes watch native stand-up comedy. Other respondents are defined as the unfamiliar group. We define familiar and unfamiliar groups for foreign type analogously. Figure \ref{talk:native} shows for native type, Tvd-variety correctly suggests that familiar group has a higher score for all 4 questions while baseline fails in three questions. Figure \ref{talk:foreign} shows for foreign type, both Tvd-variety and baseline correctly suggest that familiar group has a higher score for three questions and fail in one question. 

\input{case_talkshow/native}
\input{case_talkshow/foreign}

\paragraph{Cross questions}
We pick a group of respondents who often or sometimes watch native stand-up comedy but occasionally or almost never watch foreign stand-up comedy. The size of this group is 321. For this group, the comparisons for foreign comedians are much more difficult. We compute the Tvd-variety and baseline score of this group of respondents. Figure \ref{talk:crossq} shows that Tvd-variety successfully separates easy questions (comparisons between native comedians) and hard questions (comparisons between foreign comedians) while baseline assigns an easy question (native3) a lower score than a hard question (foreign3). 

\input{case_talkshow/crossquestion}

%\yk{add a line in the below figure as well, show that native3's baseline score is higher than foreign3}

Our results validate the advantage of our metric, compared to the baseline metric. To check the robustness of the results, we also reduce the effect of the group size by sampling (without replacement) the same number of respondents from each group for each comparison. The results are consistent. We show the results in the appendix. We also perform additional comparisons between male respondents and female respondents for all case studies. We defer the results to the appendix. 

\section{Conclusion and discussion}

Our work focuses on measuring the informativeness of a group of people in subjective questions. By additionally asking for respondents' predictions about other people's choices, we provide a refined definition of uninformative feedback. For the new definition, we propose a new family of informativeness metric, $f$-variety, for a group of people's feedback. $f$-variety separates informative and uninformative feedback and decreases as the ratio of uninformative feedback increases. We validate our metric both theoretically and empirically. 

%to solve the problem, which has the desired properties of {\bf Monotonicity} and {\bf Separation} verified by our theoretical model and numerical experiments. Additionally, we use $f$-variety for real-world experiments. \lu{In the experiments, $f$-variety performed well. On the one hand, it could distinguish between informative and uninformative groups, and on the other hand, it gave a reasonable ranking of the difficulty of the problem for the same group of people.}

Our method provides only group-level measurements. A future direction is to separate experts and non-experts in the mixed group with additional assumptions. Another future direction to theoretically explore the effect of different convex functions used for $f$-variety and further define an optimization goal and optimize over the convex functions. 

Our experimental setting focuses on binary choices, while our theory is applicable in the non-binary case. However, in practice, the respondents require additional effort to provide a prediction over non-binary choices. In the future, we can design a more practical approach in the non-binary setting. For example, one potential solution is to ask respondents for a prediction for a single choice randomly and combine them afterward. Moreover, we use flat payment in our experiments. In the future, we can consider incentives and compare the subjective feedback collected by different payment schemes.

%A future direction will be to incentivize respondents to pay more efforts in answering the questions. For example, we might pay them according to their prediction accuracy.

\bibliographystyle{plainnat}
\bibliography{reference}

% Appendix
\clearpage
\appendix
\section{More comparisons}

In this section, we show the comparison between the male group and female group in our case studies. For athletes comparisons, among all respondents, 306 are male and 350 are female. Figure \ref{ath:sex} shows the comparison results. For stand-up comedians comparisons, among all respondents, 262 are male and 370 are female. Figure \ref{talk:sex} shows the comparison results. 
\input{case_athletes/sex}
\input{case_talkshow/sex}

\clearpage
\section{Pearson-variety \& Hellinger-variety for numerical experiments}

In this section, we will additionally evaluate two special $f$-variety, Pearson-variety and Hellinger-variety by numerical experiments. Pearson-variety uses Pearson-divergence. Pearson-divergence $D_{pearson}(\mathbf{p},\mathbf{q})=\sum_{\sigma\in\Sigma}\frac{\left(p(\sigma)-q(\sigma)\right)^2}{q(\sigma)}$. Hellinger-variety uses squared Hellinger distance. Squared Hellinger distance $D_{hellinger}(\mathbf{p},\mathbf{q})=\frac{1}{2}\sum_{\sigma\in\Sigma}\left(\sqrt{p(\sigma)}-\sqrt{q(\sigma)}\right)^2$.

\input{numerical_fig/pearson-variety}
\input{numerical_fig/hellinger-variety}

\clearpage
\section{Robustness-check for group size}
In this section, for each comparison between two groups of respondents, to reduce the effect of group size, we sample the same amount of respondents without replacement from the larger size group and then compare the two equal-size groups by Tvd-variety and baseline. For example, if group A has 300 respondents and group B has 200 respondents, then we will sample 200 respondents from group A without replacement and compare group B with the subset of group A. We can compute the error bar, i.e., the standard deviation by repeating the sampling process. The following figures show the results. Due to our sampling process, only one-side has an error bar. The results still show that compared to baseline, Tvd-variety is more consistent with the reference. 

\input{case_athletes/watch-sample}
\input{case_talkshow/native-sample}
\input{case_talkshow/foreign-sample}

\clearpage
\section{Contents of surveys}
In this section, we list the questions we used in case studies. For each respondent, the orders of options are randomly shuffled.
\subsection{Survey for athletes}
\input{case_athletes/questionnaire}
\subsection{Survey for stand-up comedians}
\input{case_talkshow/questionnaire}

\end{document}

%% file: fig/tvd.tex
\begin{figure}[!h]
\begin{center}
\subfigure
{
\begin{tikzpicture}
\sffamily
\begin{axis}[
legend style={font=\small,
 nodes={scale=1, transform shape},
 at={(.03,.93)},
 anchor=north west,
 draw=none},
legend cell align={left},
width = 2.7in, height = 1.8in,
ylabel near ticks,
ytick={0, .2, .4, .6, .8, 1., 1.2, 1.4, 1.6, 1.8, 2},
ylabel = {\small Probability density},
xlabel near ticks,
every tick label/.append style={font=\scriptsize},
xmin=-.05,xmax=1.05,ymin=0,ymax=1.8,
xtick={0, .1, .2, .3, .4, .5, .6, .7, .8, .9, 1},
xlabel={\small Prediction},
xlabel style = {yshift=0.05in},
yticklabel style={
  /pgf/number format/fixed,
  /pgf/number format/precision=5
},
no markers,
legend image code/.code={
        \draw [#1] (0cm,-0.1cm) rectangle (0.2cm,0.1cm); },
]
    \addplot gnuplot [raw gnuplot,fill,opacity=.7,color=red,solid]
    {
        \Defs
        plot [x=0:1] beta(x,8,3)*.5;
    };
    \addlegendentry{choice $+$}
    \addplot gnuplot [raw gnuplot,fill,opacity=.7,color=cyan,solid]
    {
        \Defs
        plot [x=0:1] beta(x,4,5)*.5;
    };
    \addlegendentry{choice $-$}
\end{axis}
\end{tikzpicture}
}
\subfigure
{
\begin{tikzpicture}
\sffamily
\begin{axis}[
legend style={font=\small,
 nodes={scale=1, transform shape},
 at={(.03,.93)},
 anchor=north west,
 draw=none},
legend cell align={left},
width = 2.7in, height = 1.8in,
ylabel near ticks,
ytick={0, .2, .4, .6, .8, 1., 1.2, 1.4, 1.6, 1.8, 2},
ylabel = {\small Probability density},
xlabel near ticks,
every tick label/.append style={font=\scriptsize},
xmin=-.05,xmax=1.05,ymin=0,ymax=1.8,
xtick={0, .1, .2, .3, .4, .5, .6, .7, .8, .9, 1},
xlabel={\small Prediction},
xlabel style = {yshift=0.05in},
yticklabel style={
  /pgf/number format/fixed,
  /pgf/number format/precision=5
},
no markers,
legend image code/.code={
        \draw [#1] (0cm,-0.1cm) rectangle (0.2cm,0.1cm); },
]
    \addplot gnuplot [raw gnuplot,color=gray,solid,fill]
    {
        \Defs
        plot [x=0:1] max(beta(x,8,3)*.5,beta(x,4,5)*.5);
    };
    \addlegendentry{Tvd variety}
    
    \addplot gnuplot [raw gnuplot,color=white,solid,fill]
    {
        \Defs
        plot [x=0:1] min(beta(x,8,3)*.5,beta(x,4,5)*.5);
    };
\end{axis}
\end{tikzpicture}
}
\end{center}

\caption{\small \textbf{Tvd-variety in binary case:} In the binary case where choices are $\{+,-\}$, we draw the joint distribution over the choice-prediction pairs. Specifically, the blue region is the distribution over ``$-$'' people's predictions about what percentage of people will choose ``$+$'', \emph{multiplying the ratio of ``$-$'' people}. The area of the blue region is the ratio of ``$-$'' people. We plot the red region analogously. The area of the shading region, which is the difference between the blue region and red region, is proportional to the Tvd-variety. Tvd-variety is always non-negative. The above figure shows an example of ``equal affection''. In the ``random selection'' case, since choices are uniform and predictions are independent of choices, the blue region will be the same as the red region. This leads to a zero Tvd-variety.} \label{fig:tvd}
\end{figure}
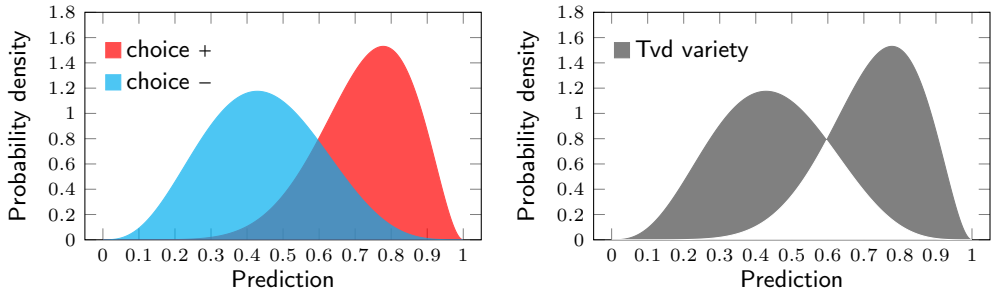

%$|q_+ \mathbf{q}_+ - q_- \mathbf{q}_-|_1$

%% file: fig/groups.tex
\begin{figure}[!h]
\centering
\subfigure[\small {\bf Dist 1 (informative)}: non-uniform choices, independent choice and prediction pairs]
{
\begin{tikzpicture}
\sffamily
\begin{axis}[
legend style={font=\small,
 nodes={scale=1, transform shape},
 at={(.03,.93)},
 anchor=north west,
 draw=none},
legend cell align={left},
width = 2.7in, height = 1.8in,
ylabel near ticks,
ytick={0, .2, .4, .6, .8, 1., 1.2, 1.4, 1.6, 1.8},
xlabel near ticks,
every tick label/.append style={font=\scriptsize},
xmin=-.05,xmax=1.05,ymin=0,ymax=1.8,
xtick={0, .1, .2, .3, .4, .5, .6, .7, .8, .9, 1},
xlabel style = {yshift=0.05in},
yticklabel style={
  /pgf/number format/fixed,
  /pgf/number format/precision=5
},
no markers,
legend image code/.code={
        \draw [#1] (0cm,-0.1cm) rectangle (0.2cm,0.1cm); },
]
    \addplot gnuplot [raw gnuplot,fill,opacity=.7,color=red,solid]
    {
        \Defs
        plot [x=0:1] beta(x,3,2)*.7;
    };
    \addlegendentry{choice $+$}
    \addplot gnuplot [raw gnuplot,fill,opacity=.7,color=cyan,solid]
    {
        \Defs
        plot [x=0:1] beta(x,3,2)*.3;
    };
    \addlegendentry{choice $-$}
\end{axis}
\end{tikzpicture}
}
\hspace*{0.5cm}
\subfigure[\small {\bf Dist 2 (informative)}: uniform choices, dependent choice and prediction pairs]
{
\begin{tikzpicture}
\sffamily
\begin{axis}[
legend style={font=\small,
 nodes={scale=1, transform shape},
 at={(.03,.93)},
 anchor=north west,
 draw=none},
legend cell align={left},
width = 2.7in, height = 1.8in,
ylabel near ticks,
ytick={0, .2, .4, .6, .8, 1., 1.2, 1.4, 1.6, 1.8},
xlabel near ticks,
every tick label/.append style={font=\scriptsize},
xmin=-.05,xmax=1.05,ymin=0,ymax=1.8,
xtick={0, .1, .2, .3, .4, .5, .6, .7, .8, .9, 1},
xlabel style = {yshift=0.05in},
yticklabel style={
  /pgf/number format/fixed,
  /pgf/number format/precision=5
},
no markers,
legend image code/.code={
        \draw [#1] (0cm,-0.1cm) rectangle (0.2cm,0.1cm); },
]
    \addplot gnuplot [raw gnuplot,fill,opacity=.7,color=red,solid]
    {
        \Defs
        plot [x=0:1] beta(x,8,3)*.5;
    };
    \addlegendentry{choice $+$}
    \addplot gnuplot [raw gnuplot,fill,opacity=.7,color=cyan,solid]
    {
        \Defs
        plot [x=0:1] beta(x,4,5)*.5;
    };
    \addlegendentry{choice $-$}
\end{axis}
\end{tikzpicture}
}
\vfill
\subfigure[\small {\bf Dist 3 (uninformative)}: uniform choices, independent choice and prediction pairs]
{
\begin{tikzpicture}
\sffamily
\begin{axis}[
legend style={font=\small,
 nodes={scale=1, transform shape},
 at={(.03,.93)},
 anchor=north west,
 draw=none},
legend cell align={left},
width = 2.7in, height = 1.8in,
ylabel near ticks,
ytick={0, .2, .4, .6, .8, 1., 1.2, 1.4, 1.6, 1.8},
xlabel near ticks,
every tick label/.append style={font=\scriptsize},
xmin=-.05,xmax=1.05,ymin=0,ymax=1.8,
xtick={0, .1, .2, .3, .4, .5, .6, .7, .8, .9, 1},
xlabel style = {yshift=0.05in},
yticklabel style={
  /pgf/number format/fixed,
  /pgf/number format/precision=5
},
no markers,
legend image code/.code={
        \draw [#1] (0cm,-0.1cm) rectangle (0.2cm,0.1cm); },
]
    \addplot gnuplot [raw gnuplot,fill,opacity=.7,color=red,solid]
    {
        \Defs
        plot [x=0:1] beta(x,3,4)*.5;
    };
    \addlegendentry{choice $+$}
    \addplot gnuplot [raw gnuplot,fill,opacity=.7,color=cyan,solid]
    {
        \Defs
        plot [x=0:1] beta(x,3,4)*.5;
    };
    \addlegendentry{choice $-$}
\end{axis}
\end{tikzpicture}
}
\hspace*{0.5cm}
\subfigure[\small {\bf Dist 4 (uninformative)}: uniform choices, independent choice and prediction pairs]
{
\begin{tikzpicture}
\sffamily
\begin{axis}[
legend style={font=\small,
 nodes={scale=1, transform shape},
 at={(.03,.93)},
 anchor=north west,
 draw=none},
legend cell align={left},
width = 2.7in, height = 1.8in,
ylabel near ticks,
ytick={0, .2, .4, .6, .8, 1., 1.2, 1.4, 1.6, 1.8},
xlabel near ticks,
every tick label/.append style={font=\scriptsize},
xmin=-.05,xmax=1.05,ymin=0,ymax=1.8,
xtick={0, .1, .2, .3, .4, .5, .6, .7, .8, .9, 1},
xlabel style = {yshift=0.05in},
yticklabel style={
  /pgf/number format/fixed,
  /pgf/number format/precision=5
},
no markers,
legend image code/.code={
        \draw [#1] (0cm,-0.1cm) rectangle (0.2cm,0.1cm); },
]
    \addplot gnuplot [raw gnuplot,fill,opacity=.7,color=red,solid]
    {
        \Defs
        plot [x=0:1] beta(x,2,2)*.5;
    };
    \addlegendentry{choice $+$}
    \addplot gnuplot [raw gnuplot,fill,opacity=.7,color=cyan,solid]
    {
        \Defs
        plot [x=0:1] beta(x,2,2)*.5;
    };
    \addlegendentry{choice $-$}
\end{axis}
\end{tikzpicture}
}

\captionof{figure}{\small \textbf{Examples of (un)informative distributions in binary case:} In distribution 1, choices and predictions are independent and the marginal distribution over choices is non-uniform, thus the blue region and red region are proportional. In distribution 2, predictions depend on choices and choices are uniform. In distributions 3 and 4, predictions and choices are independent and choices are uniform, thus the blue region and the red region are the same. In our definition, distributions 1 and 2 are informative while distributions 3 and 4 are uninformative.}
\label{fig:groups}
\end{figure}
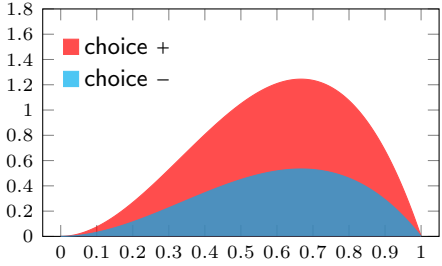
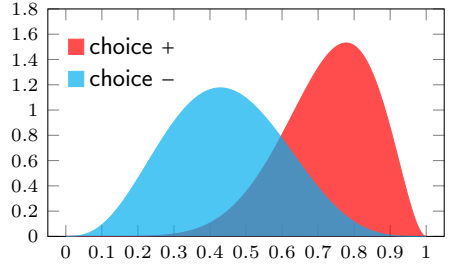
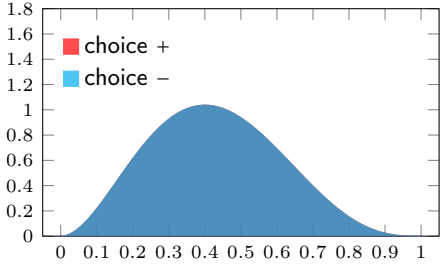
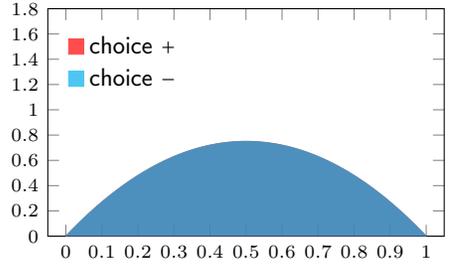

%% file: fig/discrete.tex
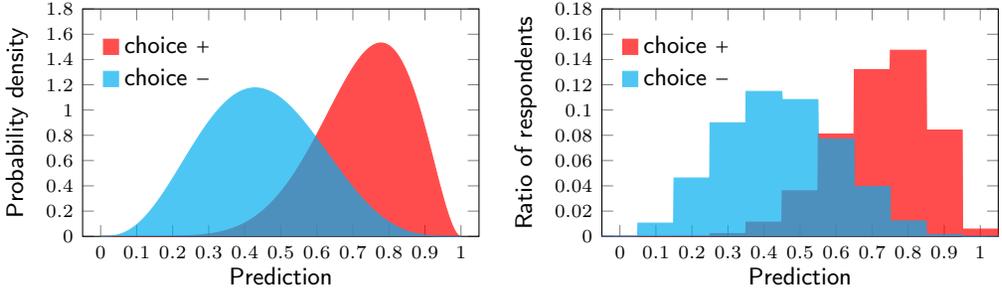
\begin{figure}[!h]
\begin{center}
\subfigure
{
\begin{tikzpicture}
\sffamily
\begin{axis}[
legend style={font=\small,
 nodes={scale=1, transform shape},
 at={(.03,.93)},
 anchor=north west,
 draw=none},
legend cell align={left},
width = 2.7in, height = 1.8in,
ylabel near ticks,
ytick={0, .2, .4, .6, .8, 1., 1.2, 1.4, 1.6, 1.8},
ylabel = {\small Probability density},
xlabel near ticks,
every tick label/.append style={font=\scriptsize},
xmin=-.05,xmax=1.05,ymin=0,ymax=1.8,
xtick={0, .1, .2, .3, .4, .5, .6, .7, .8, .9, 1},
xlabel={\small Prediction},
xlabel style = {yshift=0.05in},
yticklabel style={
  /pgf/number format/fixed,
  /pgf/number format/precision=5
},
no markers,
legend image code/.code={
        \draw [#1] (0cm,-0.1cm) rectangle (0.2cm,0.1cm); },
]
    \addplot gnuplot [raw gnuplot,fill,opacity=.7,color=red,solid]
    {
        \Defs
        plot [x=0:1] beta(x,8,3)*.5;
    };
    \addlegendentry{choice $+$}
    \addplot gnuplot [raw gnuplot,fill,opacity=.7,color=cyan,solid]
    {
        \Defs
        plot [x=0:1] beta(x,4,5)*.5;
    };
    \addlegendentry{choice $-$}
\end{axis}
\end{tikzpicture}
}
\subfigure
{
\begin{tikzpicture}
\sffamily
\begin{axis}[
legend style={font=\small,
 nodes={scale=1, transform shape},
 at={(.03,.93)},
 anchor=north west,
 draw=none},
legend cell align={left},
width = 2.7in, height = 1.8in,
ylabel near ticks,
ytick={0, .02, .04, .06, .08, 0.1, .12, .14, .16, .18},
ylabel = {\small Ratio of respondents},
xlabel near ticks,
every tick label/.append style={font=\scriptsize},
xmin=-.05,xmax=1.05,ymin=0,ymax=.18,
xtick={0, .1, .2, .3, .4, .5, .6, .7, .8, .9, 1},
xlabel={\small Prediction},
xlabel style = {yshift=0.05in},
yticklabel style={
  /pgf/number format/fixed,
  /pgf/number format/precision=5
},
no markers,
legend image code/.code={
        \draw [#1] (0cm,-0.1cm) rectangle (0.2cm,0.1cm); },
]
    \addplot gnuplot [raw gnuplot,fill,opacity=.7,color=red,solid]
    {
        \Defs
        plot [x=-.05:1.05] 
        const(x,-.05,.05,0)+
        const(x,0.05,0.15,0.00000433176)+
        const(x,0.15,0.25,0.000203568)+
        const(x,0.25,0.35,0.00220273)+
        const(x,0.35,0.45,0.0112853)+
        const(x,0.45,0.55,0.0360839)+
        const(x,0.55,0.65,0.0810239)+
        const(x,0.65,0.75,0.131993)+
        const(x,0.75,0.85,0.147302)+
        const(x,0.85,0.95,0.08415)+
        const(x,0.95,1.05,0.00575178);
    };
    \addlegendentry{choice $+$}
    \addplot gnuplot [raw gnuplot,fill,opacity=.7,color=cyan,solid]
    {
        \Defs
        plot [x=-.05:1.05] 
        const(x,-.05,.05,0.000185876)+
        const(x,0.05,0.15,0.0104904)+
        const(x,0.15,0.25,0.0462314)+
        const(x,0.25,0.35,0.0898926)+
        const(x,0.35,0.45,0.114722)+
        const(x,0.45,0.55,0.108288)+
        const(x,0.55,0.65,0.0771449)+
        const(x,0.65,0.75,0.0393965)+
        const(x,0.75,0.85,0.0122221)+
        const(x,0.85,0.95,0.00141923)+
        const(x,0.95,1.05,0.00000770244);
    };
    \addlegendentry{choice $-$}
\end{axis}
\end{tikzpicture}
}
\end{center}
\caption{\small \textbf{Empirical histogram:} The subgraph on the left is the true underlying distribution. In practice, we use histograms over finite prediction options to estimate the joint distribution, which is shown on the right.} \label{fig:discrete}
\end{figure}

%% file: numerical_fig/tvd-variety.tex
\begin{figure}[htbp]
    \begin{minipage}[b]{0.48\textwidth}
        \centering
        \input{numerical_fig/balancebeta1}
        \captionof*{figure}{\small (a) {\bf Uniform-1:} For experts, $q_+=q_-=0.5$, $\mathbf{q_+}$ is Beta(8,3) and $\mathbf{q_-}$ is Beta(4,5).}
    \end{minipage}
    \hfill
    \begin{minipage}[b]{0.48\textwidth}
        \centering
        \input{numerical_fig/unbalancebeta1}
        \captionof*{figure}{\small (b) {\bf Non-uniform-1:} For experts, $q_+=0.3$, $q_-=0.7$, $\mathbf{q_+}$ is Beta(8,3) and $\mathbf{q_-}$ is Beta(4,5).}
    \end{minipage}
    \vfill

    \begin{minipage}[b]{0.48\textwidth}
        \centering
        \input{numerical_fig/balancebeta2}
        \captionof*{figure}{\small (c) {\bf Uniform-2:} For experts, $q_+=q_-=0.5$, $\mathbf{q_+}$ is Beta(6,6) and $\mathbf{q_-}$ is Beta(2,3).}
    \end{minipage}
    \hfill
    \begin{minipage}[b]{0.48\textwidth}
        \centering
        \input{numerical_fig/unbalancebeta2}
        \captionof*{figure}{\small (d) {\bf Non-uniform-2:} For experts, $q_+=0.3$, $q_-=0.7$, $\mathbf{q_+}$ is Beta(6,6) and $\mathbf{q_-}$ is Beta(2,3).}
    \end{minipage}
    
    \captionof{figure}{\small {\bf Tvd-variety {\sl v.s.} ratio of non-experts:} In all cases, we choose Beta(2,3) as the distribution over non-experts’ prediction. In the first column (a) and (c), the choice of experts are uniform, while in the second column (b) and (d), $q_+=0.3$ and $q_-=0.7$. In the first row (a) and (b), $\mathbf{q}_+$ is Beta(8,3) and $\mathbf{q}_-$ is Beta(4,5) while in the second row (c) and (d), $\mathbf{q}_+$ is Beta(6,6) and $\mathbf{q}_-$ is Beta(2,3). In each figure, the charts below show the joint distribution at different ratios of non-experts. The symmetric difference vanishes when the ratio is one, {\sl i.e.}, there are only non-experts. The above line charts show that Tvd-variety decreases with the ratio of non-experts which verifies the {\bf Monotonicity}. When the number of samples increases, Tvd-variety almost vanishes when there are only non-experts, which verifies the {\bf Separation} property. As the sample size increases, the empirical value becomes closer to the theoretical value (black line) which is calculated by the perfect information about the joint distribution.}
    \label{fig:numfig}
\end{figure}

%% file: numerical_fig/balancebeta1.tex
\begin{tikzpicture}[scale=.5]
\sffamily
\begin{axis}[
title style={align=center,yshift=-.0in},
legend style={font=\large,
 nodes={scale=1, transform shape},
 at={(.7,.9)},
 anchor=north west,
 draw=none},
legend cell align={left},
width = 4.5in, height = 3in,
ylabel near ticks,
ytick={0, .05, .1, .15, .2, .25, .3, .35, .4},
ylabel = {Tvd-variety},
xlabel near ticks,
every tick label/.append style={font=\normalsize},
xmin=-0.05,xmax=1.05,ymin=0,ymax=.4,
xtick={0, .1, .2, .3, .4, .5, .6, .7, .8, .9, 1},
xlabel={ratio of uninformative participants},
xlabel style = {font=\Large},
ylabel style = {font=\Large},
yticklabel style={
  /pgf/number format/fixed,
  /pgf/number format/precision=5
},
%scaled y ticks=false
%every axis plot/.append style={thick}
]

\addplot[solid, mark=., mark options={scale=.8}, gray, style=thick,]
table[]{
0. 0.329576
1. 0.
};
\addlegendentry{theoretical}

\addplot[solid, mark=o, mark options={scale=.8}, red]
plot [error bars/.cd, y dir = both, y explicit]
table[y error index=2]{
0.0030 0.3231 0.0115
0.1030 0.2917 0.0120
0.2030 0.2601 0.0118
0.3030 0.2279 0.0125
0.4030 0.1960 0.0126
0.5030 0.1658 0.0132
0.6030 0.1343 0.0132
0.7030 0.1036 0.0138
0.8030 0.0738 0.0131
0.9030 0.0493 0.0115
1.0030 0.0384 0.0095
};
\addlegendentry{$n=1000$}

\addplot[solid, mark=o, mark options={scale=.8}, orange]
plot [error bars/.cd, y dir = both, y explicit]
table[y error index=2]{
0.0010 0.3256 0.0158
0.1010 0.2927 0.0169
0.2010 0.2625 0.0172
0.3010 0.2311 0.0175
0.4010 0.2007 0.0178
0.5010 0.1691 0.0181
0.6010 0.1376 0.0183
0.7010 0.1087 0.0176
0.8010 0.0827 0.0174
0.9010 0.0628 0.0147
1.0010 0.0548 0.0135
};
\addlegendentry{$n=500$}

\addplot[solid, mark=o, mark options={scale=.8}, green]
plot [error bars/.cd, y dir = both, y explicit]
table[y error index=2]{
-0.0010 0.3300 0.0241
0.0990 0.2988 0.0244
0.1990 0.2681 0.0268
0.2990 0.2389 0.0277
0.3990 0.2082 0.0272
0.4990 0.1769 0.0275
0.5990 0.1504 0.0276
0.6990 0.1251 0.0254
0.7990 0.1046 0.0239
0.8990 0.0914 0.0232
0.9990 0.0862 0.0213
};
\addlegendentry{$n=200$}

\addplot[solid, mark=o, mark options={scale=.8}, cyan]
plot [error bars/.cd, y dir = both, y explicit]
table[y error index=2]{
-0.0030 0.3340 0.0343
0.0970 0.3026 0.0343
0.1970 0.2758 0.0362
0.2970 0.2460 0.0374
0.3970 0.2199 0.0366
0.4970 0.1921 0.0344
0.5970 0.1679 0.0341
0.6970 0.1480 0.0336
0.7970 0.1332 0.0326
0.8970 0.1235 0.0306
0.9970 0.1204 0.0312
};
\addlegendentry{$n=100$}

\end{axis}
\end{tikzpicture}

\begin{tikzpicture}[scale=.5]
\sffamily
\begin{axis}[
width = 5.47in, height = 1.5in,
ylabel near ticks,
ytick=\empty,
ylabel = {\small Probability Density},
%xlabel near ticks,
%every tick label/.append style={font=\tiny},
xmin=-.84,xmax=6,ymin=0,ymax=1.8,
xtick=\empty,
yticklabel style={
  /pgf/number format/fixed,
  /pgf/number format/precision=5
},
no markers,
hide axis,
]
    \addplot gnuplot [raw gnuplot,fill,opacity=.7,color=red,solid]
    {
        \Defs
        plot [x=0:1] beta(x,8,3)*.5;
    };
    \addplot gnuplot [raw gnuplot,fill,opacity=.7,color=cyan,solid]
    {
        \Defs
        plot [x=0:1] beta(x,4,5)*.5;
    };

    \addplot gnuplot [raw gnuplot,fill,opacity=.7,color=red,solid]
    {
        \Defs
        plot [x=1:2] (.8*beta(x-1,8,3)+.2*beta(x-1,2,2))*.5;
    };
    \addplot gnuplot [raw gnuplot,fill,opacity=.7,color=cyan,solid]
    {
        \Defs
        plot [x=1:2] (.8*beta(x-1,4,5)+.2*beta(x-1,2,2))*.5;
    };

    \addplot gnuplot [raw gnuplot,fill,opacity=.7,color=red,solid]
    {
        \Defs
        plot [x=2:3] (.6*beta(x-2,8,3)+.4*beta(x-2,2,2))*.5;
    };
    \addplot gnuplot [raw gnuplot,fill,opacity=.7,color=cyan,solid]
    {
        \Defs
        plot [x=2:3] (.6*beta(x-2,4,5)+.4*beta(x-2,2,2))*.5;
    };

    \addplot gnuplot [raw gnuplot,fill,opacity=.7,color=red,solid]
    {
        \Defs
        plot [x=3:4] (.4*beta(x-3,8,3)+.6*beta(x-3,2,2))*.5;
    };
    \addplot gnuplot [raw gnuplot,fill,opacity=.7,color=cyan,solid]
    {
        \Defs
        plot [x=3:4] (.4*beta(x-3,4,5)+.6*beta(x-3,2,2))*.5;
    };

    \addplot gnuplot [raw gnuplot,fill,opacity=.7,color=red,solid]
    {
        \Defs
        plot [x=4:5] (.2*beta(x-4,8,3)+.8*beta(x-4,2,2))*.5;
    };
    \addplot gnuplot [raw gnuplot,fill,opacity=.7,color=cyan,solid]
    {
        \Defs
        plot [x=4:5] (.2*beta(x-4,4,5)+.8*beta(x-4,2,2))*.5;
    };

    \addplot gnuplot [raw gnuplot,fill,opacity=.7,color=red,solid]
    {
        \Defs
        plot [x=5:6] beta(x-5,2,2)*.5;
    };
    \addplot gnuplot [raw gnuplot,fill,opacity=.7,color=cyan,solid]
    {
        \Defs
        plot [x=5:6] beta(x-5,2,2)*.5;
    };

    \end{axis}
\end{tikzpicture}

%% file: numerical_fig/unbalancebeta1.tex
\begin{tikzpicture}[scale=.5]
\sffamily
\begin{axis}[
title style={align=center,yshift=-.0in},
legend style={font=\large,
 nodes={scale=1, transform shape},
 at={(.7,.9)},
 anchor=north west,
 draw=none},
legend cell align={left},
width = 4.5in, height = 3in,
ylabel near ticks,
ytick={0, .05, .1, .15, .2, .25, .3, .35, .4},
ylabel = {Tvd-variety},
xlabel near ticks,
every tick label/.append style={font=\normalsize},
xmin=-0.05,xmax=1.05,ymin=0,ymax=.4,
xtick={0, .1, .2, .3, .4, .5, .6, .7, .8, .9, 1},
xlabel={ratio of uninformative participants},
xlabel style = {font=\Large},
ylabel style = {font=\Large},
yticklabel style={
  /pgf/number format/fixed,
  /pgf/number format/precision=5
},
%scaled y ticks=false
%every axis plot/.append style={thick}
]

\addplot[solid, mark=., mark options={scale=.8}, gray, style=thick,]
table[]{
0. 0.348710
1. 0.
};
\addlegendentry{theoretical}

\addplot[solid, mark=o, mark options={scale=.8}, red]
plot [error bars/.cd, y dir = both, y explicit]
table[y error index=2]{
0.0030 0.3471 0.0115
0.1030 0.3133 0.0120
0.2030 0.2783 0.0123
0.3030 0.2440 0.0123
0.4030 0.2097 0.0129
0.5030 0.1764 0.0132
0.6030 0.1424 0.0135
0.7030 0.1095 0.0134
0.8030 0.0762 0.0132
0.9030 0.0502 0.0117
1.0030 0.0396 0.0095
};
\addlegendentry{$n=1000$}

\addplot[solid, mark=o, mark options={scale=.8}, orange]
plot [error bars/.cd, y dir = both, y explicit]
table[y error index=2]{
0.0010 0.3473 0.0157
0.1010 0.3140 0.0164
0.2010 0.2792 0.0168
0.3010 0.2459 0.0174
0.4010 0.2112 0.0178
0.5010 0.1784 0.0186
0.6010 0.1457 0.0193
0.7010 0.1137 0.0184
0.8010 0.0859 0.0175
0.9010 0.0638 0.0149
1.0010 0.0546 0.0138
};
\addlegendentry{$n=500$}

\addplot[solid, mark=o, mark options={scale=.8}, green]
plot [error bars/.cd, y dir = both, y explicit]
table[y error index=2]{
-0.0010 0.3501 0.0241
0.0990 0.3175 0.0245
0.1990 0.2838 0.0253
0.2990 0.2513 0.0263
0.3990 0.2185 0.0282
0.4990 0.1879 0.0283
0.5990 0.1561 0.0273
0.6990 0.1280 0.0257
0.7990 0.1062 0.0250
0.8990 0.0916 0.0220
0.9990 0.0883 0.0215
};
\addlegendentry{$n=200$}

\addplot[solid, mark=o, mark options={scale=.8}, cyan]
plot [error bars/.cd, y dir = both, y explicit]
table[y error index=2]{
-0.0030 0.3529 0.0333
0.0970 0.3218 0.0334
0.1970 0.2913 0.0350
0.2970 0.2581 0.0360
0.3970 0.2293 0.0361
0.4970 0.1994 0.0380
0.5970 0.1734 0.0367
0.6970 0.1530 0.0344
0.7970 0.1339 0.0309
0.8970 0.1255 0.0311
0.9970 0.1192 0.0301
};
\addlegendentry{$n=100$}

\end{axis}
\end{tikzpicture}

\begin{tikzpicture}[scale=.5]
\sffamily
\begin{axis}[
width = 5.47in, height = 1.5in,
ylabel near ticks,
ytick=\empty,
ylabel = {\small Probability Density},
%xlabel near ticks,
%every tick label/.append style={font=\tiny},
xmin=-.84,xmax=6,ymin=0,ymax=1.8,
xtick=\empty,
yticklabel style={
  /pgf/number format/fixed,
  /pgf/number format/precision=5
},
no markers,
hide axis,
]
    \addplot gnuplot [raw gnuplot,fill,opacity=.7,color=red,solid]
    {
        \Defs
        plot [x=0:1] beta(x,8,3)*.3;
    };
    \addplot gnuplot [raw gnuplot,fill,opacity=.7,color=cyan,solid]
    {
        \Defs
        plot [x=0:1] beta(x,4,5)*.7;
    };

    \addplot gnuplot [raw gnuplot,fill,opacity=.7,color=red,solid]
    {
        \Defs
        plot [x=1:2] (.8*beta(x-1,8,3)*.3+.2*beta(x-1,2,2)*.5);
    };
    \addplot gnuplot [raw gnuplot,fill,opacity=.7,color=cyan,solid]
    {
        \Defs
        plot [x=1:2] (.8*beta(x-1,4,5)*.7+.2*beta(x-1,2,2)*.5);
    };

    \addplot gnuplot [raw gnuplot,fill,opacity=.7,color=red,solid]
    {
        \Defs
        plot [x=2:3] (.6*beta(x-2,8,3)*.3+.4*beta(x-2,2,2)*.5);
    };
    \addplot gnuplot [raw gnuplot,fill,opacity=.7,color=cyan,solid]
    {
        \Defs
        plot [x=2:3] (.6*beta(x-2,4,5)*.7+.4*beta(x-2,2,2)*.5);
    };

    \addplot gnuplot [raw gnuplot,fill,opacity=.7,color=red,solid]
    {
        \Defs
        plot [x=3:4] (.4*beta(x-3,8,3)*.3+.6*beta(x-3,2,2)*.5);
    };
    \addplot gnuplot [raw gnuplot,fill,opacity=.7,color=cyan,solid]
    {
        \Defs
        plot [x=3:4] (.4*beta(x-3,4,5)*.7+.6*beta(x-3,2,2)*.5);
    };

    \addplot gnuplot [raw gnuplot,fill,opacity=.7,color=red,solid]
    {
        \Defs
        plot [x=4:5] (.2*beta(x-4,8,3)*.3+.8*beta(x-4,2,2)*.5);
    };
    \addplot gnuplot [raw gnuplot,fill,opacity=.7,color=cyan,solid]
    {
        \Defs
        plot [x=4:5] (.2*beta(x-4,4,5)*.7+.8*beta(x-4,2,2)*.5);
    };

    \addplot gnuplot [raw gnuplot,fill,opacity=.7,color=red,solid]
    {
        \Defs
        plot [x=5:6] beta(x-5,2,2)*.5;
    };
    \addplot gnuplot [raw gnuplot,fill,opacity=.7,color=cyan,solid]
    {
        \Defs
        plot [x=5:6] beta(x-5,2,2)*.5;
    };

    \end{axis}
\end{tikzpicture}

%% file: numerical_fig/balancebeta2.tex
\begin{tikzpicture}[scale=.5]
\sffamily
\begin{axis}[
title style={align=center,yshift=-.0in},
legend style={font=\large,
 nodes={scale=1, transform shape},
 at={(.7,.9)},
 anchor=north west,
 draw=none},
legend cell align={left},
width = 4.5in, height = 3in,
ylabel near ticks,
ytick={0, .05, .1, .15, .2, .25, .3, .35, .4},
ylabel = {Tvd-variety},
xlabel near ticks,
every tick label/.append style={font=\normalsize},
xmin=-0.05,xmax=1.05,ymin=0,ymax=.4,
xtick={0, .1, .2, .3, .4, .5, .6, .7, .8, .9, 1},
xlabel={ratio of uninformative participants},
xlabel style = {font=\Large},
ylabel style = {font=\Large},
yticklabel style={
  /pgf/number format/fixed,
  /pgf/number format/precision=5
},
%scaled y ticks=false
%every axis plot/.append style={thick}
]

\addplot[solid, mark=., mark options={scale=.8}, gray, style=thick,]
table[]{
0 0.152907
1. 0.
};
\addlegendentry{theoretical}

\addplot[solid, mark=o, mark options={scale=.8}, red]
plot [error bars/.cd, y dir = both, y explicit]
table[y error index=2]{
0.0030 0.1541 0.0146
0.1030 0.1391 0.0149
0.2030 0.1253 0.0140
0.3030 0.1105 0.0147
0.4030 0.0969 0.0135
0.5030 0.0832 0.0135
0.6030 0.0708 0.0131
0.7030 0.0582 0.0123
0.8030 0.0492 0.0114
0.9030 0.0413 0.0100
1.0030 0.0384 0.0097
};
\addlegendentry{$n=1000$}

\addplot[solid, mark=o, mark options={scale=.8}, orange]
plot [error bars/.cd, y dir = both, y explicit]
table[y error index=2]{
0.0010 0.1559 0.0203
0.1010 0.1424 0.0196
0.2010 0.1295 0.0200
0.3010 0.1159 0.0186
0.4010 0.1045 0.0186
0.5010 0.0919 0.0179
0.6010 0.0795 0.0171
0.7010 0.0696 0.0158
0.8010 0.0617 0.0144
0.9010 0.0576 0.0140
1.0010 0.0537 0.0130
};
\addlegendentry{$n=500$}

\addplot[solid, mark=o, mark options={scale=.8}, green]
plot [error bars/.cd, y dir = both, y explicit]
table[y error index=2]{
-0.0010 0.1638 0.0297
0.0990 0.1543 0.0282
0.1990 0.1429 0.0293
0.2990 0.1298 0.0269
0.3990 0.1201 0.0267
0.4990 0.1110 0.0263
0.5990 0.1027 0.0242
0.6990 0.0959 0.0229
0.7990 0.0911 0.0216
0.8990 0.0875 0.0220
0.9990 0.0874 0.0213
};
\addlegendentry{$n=200$}

\addplot[solid, mark=o, mark options={scale=.8}, cyan]
plot [error bars/.cd, y dir = both, y explicit]
table[y error index=2]{
-0.0030 0.1806 0.0374
0.0970 0.1691 0.0380
0.1970 0.1613 0.0358
0.2970 0.1542 0.0347
0.3970 0.1434 0.0344
0.4970 0.1389 0.0336
0.5970 0.1293 0.0336
0.6970 0.1270 0.0316
0.7970 0.1234 0.0302
0.8970 0.1208 0.0299
0.9970 0.1213 0.0302
};
\addlegendentry{$n=100$}

\end{axis}
\end{tikzpicture}

\begin{tikzpicture}[scale=.5]
\sffamily
\begin{axis}[
width = 5.47in, height = 1.5in,
ylabel near ticks,
ytick=\empty,
ylabel = {\small Probability Density},
%xlabel near ticks,
%every tick label/.append style={font=\tiny},
xmin=-.84,xmax=6,ymin=0,ymax=1.8,
xtick=\empty,
yticklabel style={
  /pgf/number format/fixed,
  /pgf/number format/precision=5
},
no markers,
hide axis,
]
    \addplot gnuplot [raw gnuplot,fill,opacity=.7,color=red,solid]
    {
        \Defs
        plot [x=0:1] beta(x,6,6)*.5;
    };
    \addplot gnuplot [raw gnuplot,fill,opacity=.7,color=cyan,solid]
    {
        \Defs
        plot [x=0:1] beta(x,2,3)*.5;
    };

    \addplot gnuplot [raw gnuplot,fill,opacity=.7,color=red,solid]
    {
        \Defs
        plot [x=1:2] (.8*beta(x-1,6,6)*.5+.2*beta(x-1,2,2)*.5);
    };
    \addplot gnuplot [raw gnuplot,fill,opacity=.7,color=cyan,solid]
    {
        \Defs
        plot [x=1:2] (.8*beta(x-1,2,3)*.5+.2*beta(x-1,2,2)*.5);
    };

    \addplot gnuplot [raw gnuplot,fill,opacity=.7,color=red,solid]
    {
        \Defs
        plot [x=2:3] (.6*beta(x-2,6,6)*.5+.4*beta(x-2,2,2)*.5);
    };
    \addplot gnuplot [raw gnuplot,fill,opacity=.7,color=cyan,solid]
    {
        \Defs
        plot [x=2:3] (.6*beta(x-2,2,3)*.5+.4*beta(x-2,2,2)*.5);
    };

    \addplot gnuplot [raw gnuplot,fill,opacity=.7,color=red,solid]
    {
        \Defs
        plot [x=3:4] (.4*beta(x-3,6,6)*.5+.6*beta(x-3,2,2)*.5);
    };
    \addplot gnuplot [raw gnuplot,fill,opacity=.7,color=cyan,solid]
    {
        \Defs
        plot [x=3:4] (.4*beta(x-3,2,3)*.5+.6*beta(x-3,2,2)*.5);
    };

    \addplot gnuplot [raw gnuplot,fill,opacity=.7,color=red,solid]
    {
        \Defs
        plot [x=4:5] (.2*beta(x-4,6,6)*.5+.8*beta(x-4,2,2)*.5);
    };
    \addplot gnuplot [raw gnuplot,fill,opacity=.7,color=cyan,solid]
    {
        \Defs
        plot [x=4:5] (.2*beta(x-4,2,3)*.5+.8*beta(x-4,2,2)*.5);
    };

    \addplot gnuplot [raw gnuplot,fill,opacity=.7,color=red,solid]
    {
        \Defs
        plot [x=5:6] beta(x-5,2,2)*.5;
    };
    \addplot gnuplot [raw gnuplot,fill,opacity=.7,color=cyan,solid]
    {
        \Defs
        plot [x=5:6] beta(x-5,2,2)*.5;
    };

    \end{axis}
\end{tikzpicture}

%% file: numerical_fig/unbalancebeta2.tex
\begin{tikzpicture}[scale=.5]
\sffamily
\begin{axis}[
title style={align=center,yshift=-.0in},
legend style={font=\large,
 nodes={scale=1, transform shape},
 at={(.7,.9)},
 anchor=north west,
 draw=none},
legend cell align={left},
width = 4.5in, height = 3in,
ylabel near ticks,
ytick={0, .05, .1, .15, .2, .25, .3, .35, .4},
ylabel = {Tvd-variety},
xlabel near ticks,
every tick label/.append style={font=\normalsize},
xmin=-0.05,xmax=1.05,ymin=0,ymax=.4,
xtick={0, .1, .2, .3, .4, .5, .6, .7, .8, .9, 1},
xlabel={ratio of uninformative participants},
xlabel style = {font=\Large},
ylabel style = {font=\Large},
yticklabel style={
  /pgf/number format/fixed,
  /pgf/number format/precision=5
},
%scaled y ticks=false
%every axis plot/.append style={thick}
]

\addplot[solid, mark=., mark options={scale=.8}, gray, style=thick,]
table[]{
0. 0.2
1. 0.
};
\addlegendentry{theoretical}

\addplot[solid, mark=o, mark options={scale=.8}, red]
plot [error bars/.cd, y dir = both, y explicit]
table[y error index=2]{
0.0030 0.2012 0.0134
0.1030 0.1815 0.0140
0.2030 0.1621 0.0137
0.3030 0.1427 0.0136
0.4030 0.1253 0.0139
0.5030 0.1061 0.0129
0.6030 0.0882 0.0129
0.7030 0.0716 0.0125
0.8030 0.0563 0.0118
0.9030 0.0437 0.0101
1.0030 0.0393 0.0100
};
\addlegendentry{$n=1000$}

\addplot[solid, mark=o, mark options={scale=.8}, orange]
plot [error bars/.cd, y dir = both, y explicit]
table[y error index=2]{
0.0010 0.2029 0.0186
0.1010 0.1847 0.0188
0.2010 0.1663 0.0191
0.3010 0.1484 0.0185
0.4010 0.1303 0.0187
0.5010 0.1139 0.0185
0.6010 0.0960 0.0169
0.7010 0.0810 0.0158
0.8010 0.0677 0.0149
0.9010 0.0584 0.0140
1.0010 0.0549 0.0131
};
\addlegendentry{$n=500$}

\addplot[solid, mark=o, mark options={scale=.8}, green]
plot [error bars/.cd, y dir = both, y explicit]
table[y error index=2]{
-0.0010 0.2120 0.0276
0.0990 0.1952 0.0278
0.1990 0.1774 0.0272
0.2990 0.1628 0.0267
0.3990 0.1456 0.0275
0.4990 0.1314 0.0258
0.5990 0.1166 0.0246
0.6990 0.1047 0.0244
0.7990 0.0948 0.0231
0.8990 0.0895 0.0213
0.9990 0.0863 0.0214
};
\addlegendentry{$n=200$}

\addplot[solid, mark=o, mark options={scale=.8}, cyan]
plot [error bars/.cd, y dir = both, y explicit]
table[y error index=2]{
-0.0030 0.2243 0.0377
0.0970 0.2092 0.0377
0.1970 0.1956 0.0373
0.2970 0.1797 0.0363
0.3970 0.1665 0.0361
0.4970 0.1560 0.0341
0.5970 0.1427 0.0341
0.6970 0.1316 0.0317
0.7970 0.1254 0.0311
0.8970 0.1226 0.0319
0.9970 0.1206 0.0296
};
\addlegendentry{$n=100$}

\end{axis}
\end{tikzpicture}

\begin{tikzpicture}[scale=.5]
\sffamily
\begin{axis}[
width = 5.47in, height = 1.5in,
ylabel near ticks,
ytick=\empty,
ylabel = {\small Probability Density},
%xlabel near ticks,
%every tick label/.append style={font=\tiny},
xmin=-.84,xmax=6,ymin=0,ymax=1.8,
xtick=\empty,
yticklabel style={
  /pgf/number format/fixed,
  /pgf/number format/precision=5
},
no markers,
hide axis,
]
    \addplot gnuplot [raw gnuplot,fill,opacity=.7,color=red,solid]
    {
        \Defs
        plot [x=0:1] beta(x,6,6)*.3;
    };
    \addplot gnuplot [raw gnuplot,fill,opacity=.7,color=cyan,solid]
    {
        \Defs
        plot [x=0:1] beta(x,2,3)*.7;
    };

    \addplot gnuplot [raw gnuplot,fill,opacity=.7,color=red,solid]
    {
        \Defs
        plot [x=1:2] (.8*beta(x-1,6,6)*.3+.2*beta(x-1,2,2)*.5);
    };
    \addplot gnuplot [raw gnuplot,fill,opacity=.7,color=cyan,solid]
    {
        \Defs
        plot [x=1:2] (.8*beta(x-1,2,3)*.7+.2*beta(x-1,2,2)*.5);
    };

    \addplot gnuplot [raw gnuplot,fill,opacity=.7,color=red,solid]
    {
        \Defs
        plot [x=2:3] (.6*beta(x-2,6,6)*.3+.4*beta(x-2,2,2)*.5);
    };
    \addplot gnuplot [raw gnuplot,fill,opacity=.7,color=cyan,solid]
    {
        \Defs
        plot [x=2:3] (.6*beta(x-2,2,3)*.7+.4*beta(x-2,2,2)*.5);
    };

    \addplot gnuplot [raw gnuplot,fill,opacity=.7,color=red,solid]
    {
        \Defs
        plot [x=3:4] (.4*beta(x-3,6,6)*.3+.6*beta(x-3,2,2)*.5);
    };
    \addplot gnuplot [raw gnuplot,fill,opacity=.7,color=cyan,solid]
    {
        \Defs
        plot [x=3:4] (.4*beta(x-3,2,3)*.7+.6*beta(x-3,2,2)*.5);
    };

    \addplot gnuplot [raw gnuplot,fill,opacity=.7,color=red,solid]
    {
        \Defs
        plot [x=4:5] (.2*beta(x-4,6,6)*.3+.8*beta(x-4,2,2)*.5);
    };
    \addplot gnuplot [raw gnuplot,fill,opacity=.7,color=cyan,solid]
    {
        \Defs
        plot [x=4:5] (.2*beta(x-4,2,3)*.7+.8*beta(x-4,2,2)*.5);
    };

    \addplot gnuplot [raw gnuplot,fill,opacity=.7,color=red,solid]
    {
        \Defs
        plot [x=5:6] beta(x-5,2,2)*.5;
    };
    \addplot gnuplot [raw gnuplot,fill,opacity=.7,color=cyan,solid]
    {
        \Defs
        plot [x=5:6] beta(x-5,2,2)*.5;
    };

    \end{axis}
\end{tikzpicture}

%% file: case_athletes/watch.tex
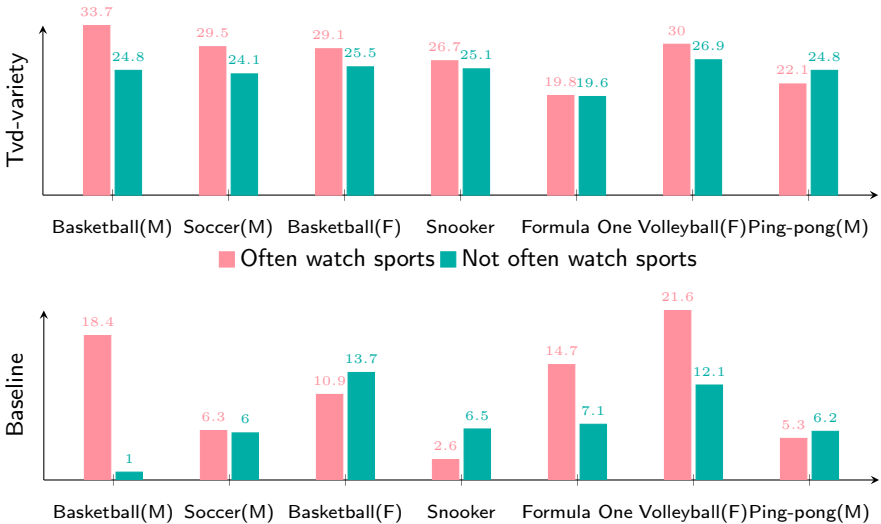
\begin{figure}[!h]
\definecolor{col1}{HTML}{ff919e}
\definecolor{col2}{HTML}{00aea5}
\centering
\begin{tikzpicture}[font=\small]
\sffamily
\begin{axis}[
  ybar,
  bar width=10pt,
  ylabel={Tvd-variety},
  ylabel near ticks,
  width = 5in, height = 1.5in,
  ymin=0,
  ytick=\empty,
  xtick=data,
  tick label style={font=\scriptsize},
  axis x line=bottom,
  axis y line=left,
  enlarge x limits=0.1,
  symbolic x coords={Basketball(M),Soccer(M),Basketball(F),Snooker,Formula One,Volleyball(F),Ping-pong(M)},
  xticklabel style={anchor=base,yshift=-\baselineskip},
  nodes near coords={\pgfmathprintnumber\pgfplotspointmeta},
  %x tick label style={font=\tiny,},
  legend image code/.code={
        \draw [#1] (0cm,-0.1cm) rectangle (0.2cm,0.1cm); },
  legend style={
    font=\small,
    nodes={scale=1, transform shape},
    at={(0.5,-0.25)},
    anchor=north,
    draw=none,
    legend columns=-1},
]
  \addplot[fill=white,bar shift=-6pt,color=col1,font=\tiny] coordinates {
(Basketball(M),33.7)
(Soccer(M),29.5)
(Basketball(F),29.1)
(Snooker,26.7)
(Formula One,19.8)
(Volleyball(F),30.0)
(Ping-pong(M),22.1)
  };
  \addlegendentry{Often watch sports}
  \addplot[fill=white,bar shift=6pt,color=col2,font=\tiny] coordinates {
(Basketball(M),24.8)
(Soccer(M),24.1)
(Basketball(F),25.5)
(Snooker,25.1)
(Formula One,19.6)
(Volleyball(F),26.9)
(Ping-pong(M),24.8)
  };
  \addlegendentry{Not often watch sports}
\end{axis}
\end{tikzpicture}
\begin{tikzpicture}[font=\small]
\sffamily
\begin{axis}[
  ybar,
  bar width=10pt,
  ylabel={\textcolor{white}{|}Baseline\textcolor{white}{|}},
  ylabel near ticks,
  width = 5in, height = 1.5in,
  ymin=0,
  ytick=\empty,
  xtick=data,
  tick label style={font=\scriptsize},
  axis x line=bottom,
  axis y line=left,
  enlarge x limits=0.1,
  symbolic x coords={Basketball(M),Soccer(M),Basketball(F),Snooker,Formula One,Volleyball(F),Ping-pong(M)},
  xticklabel style={anchor=base,yshift=-\baselineskip},
  nodes near coords={\pgfmathprintnumber\pgfplotspointmeta},
  %x tick label style={font=\tiny,},
  legend image code/.code={
        \draw [#1] (0cm,-0.1cm) rectangle (0.2cm,0.1cm); },
  legend style={
    font=\small,
    nodes={scale=1, transform shape},
    at={(0.5,-0.25)},
    anchor=north,
    draw=none,
    legend columns=-1},
]
  \addplot[fill=white,bar shift=-6pt,color=col1,font=\tiny] coordinates {
(Basketball(M),18.4)
(Soccer(M),6.3)
(Basketball(F),10.9)
(Snooker,2.6)
(Formula One,14.7)
(Volleyball(F),21.6)
(Ping-pong(M),5.3)
  };
  \addplot[fill=white,bar shift=6pt,color=col2,font=\tiny] coordinates {
(Basketball(M),1.0)
(Soccer(M),6.0)
(Basketball(F),13.7)
(Snooker,6.5)
(Formula One,7.1)
(Volleyball(F),12.1)
(Ping-pong(M),6.2)
  };
  
\end{axis}
\end{tikzpicture}
\captionof{figure}{\small {\bf Often watching sports  {\sl v.s.} not:} The questions ask for different sports. Thus, we can use the sports category (e.g. Basketball(M) means men's basketball and Basketball(F) means women's basketball) to represent the questions. The red bar represents the metrics of the group that contains people who often watch sports, while the green bar represents the metrics of the group that contains people who do not often watch sports. The y-axis is 100 times of metrics.}
\label{ath:watch}
\end{figure}

%% file: case_talkshow/native.tex
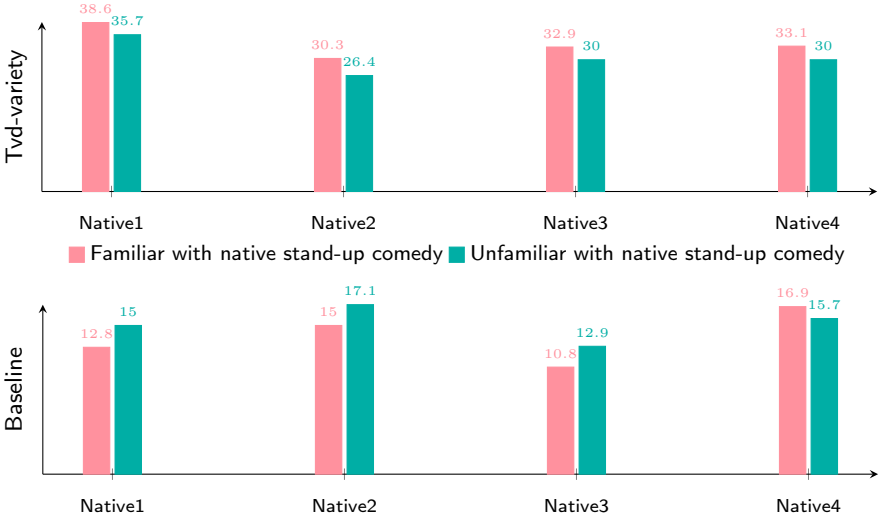
\begin{figure}[!h]
\definecolor{col1}{HTML}{ff919e}
\definecolor{col2}{HTML}{00aea5}
\centering
\begin{tikzpicture}[font=\small]
\sffamily
\begin{axis}[
  ybar,
  bar width=10pt,
  ylabel={Tvd-variety},
  ylabel near ticks,
  width = 5in, height = 1.5in,
  ymin=0,
  ytick=\empty,
  xtick=data,
  tick label style={font=\scriptsize},
  axis x line=bottom,
  axis y line=left,
  enlarge x limits=0.1,
  symbolic x coords={Native1,Native2,Native3,Native4},
  xticklabel style={anchor=base,yshift=-\baselineskip},
  nodes near coords={\pgfmathprintnumber\pgfplotspointmeta},
  %x tick label style={font=\tiny,},
  legend image code/.code={
        \draw [#1] (0cm,-0.1cm) rectangle (0.2cm,0.1cm); },
  legend style={
    font=\small,
    nodes={scale=1, transform shape},
    at={(0.5,-0.25)},
    anchor=north,
    draw=none,
    legend columns=-1},
]
  \addplot[fill=white,bar shift=-6pt,color=col1,font=\tiny] coordinates {
(Native1,38.6)
(Native2,30.3)
(Native3,32.9)
(Native4,33.1)
  };
  \addlegendentry{\footnotesize Familiar with native stand-up comedy}
  \addplot[fill=white,bar shift=6pt,color=col2,font=\tiny] coordinates {
(Native1,35.7)
(Native2,26.4)
(Native3,30.0)
(Native4,30.0)
  };
  \addlegendentry{\footnotesize Unfamiliar with native stand-up comedy}
\end{axis}
\end{tikzpicture}
\begin{tikzpicture}[font=\small]
\sffamily
\begin{axis}[
  ybar,
  bar width=10pt,
  ylabel={\textcolor{white}{|}Baseline\textcolor{white}{|}},
  ylabel near ticks,
  width = 5in, height = 1.5in,
  ymin=0,
  ytick=\empty,
  xtick=data,
  tick label style={font=\scriptsize},
  axis x line=bottom,
  axis y line=left,
  enlarge x limits=0.1,
  symbolic x coords={Native1,Native2,Native3,Native4},
  xticklabel style={anchor=base,yshift=-\baselineskip},
  nodes near coords={\pgfmathprintnumber\pgfplotspointmeta},
  %x tick label style={font=\tiny,},
  legend image code/.code={
        \draw [#1] (0cm,-0.1cm) rectangle (0.2cm,0.1cm); },
  legend style={
    font=\small,
    nodes={scale=1, transform shape},
    at={(0.5,-0.25)},
    anchor=north,
    draw=none,
    legend columns=-1},
]
  \addplot[fill=white,bar shift=-6pt,color=col1,font=\tiny] coordinates {
(Native1,12.8)
(Native2,15.0)
(Native3,10.8)
(Native4,16.9)
  };
  \addplot[fill=white,bar shift=6pt,color=col2,font=\tiny] coordinates {
(Native1,15.0)
(Native2,17.1)
(Native3,12.9)
(Native4,15.7)
  };
  
\end{axis}
\end{tikzpicture}
\captionof{figure}{\small {\bf Familiar with native stand-up comedy {\sl v.s.} unfamiliar:} There are four questions that compare two native stand-up comedians. The red bar represents the metrics of the group that contains people who are familiar with native stand-up comedy, while the green bar represents the metrics of the group who are unfamiliar with native stand-up comedy. The y-axis is 100 times of metrics.}
\label{talk:native}
\end{figure}

%% file: case_talkshow/foreign.tex
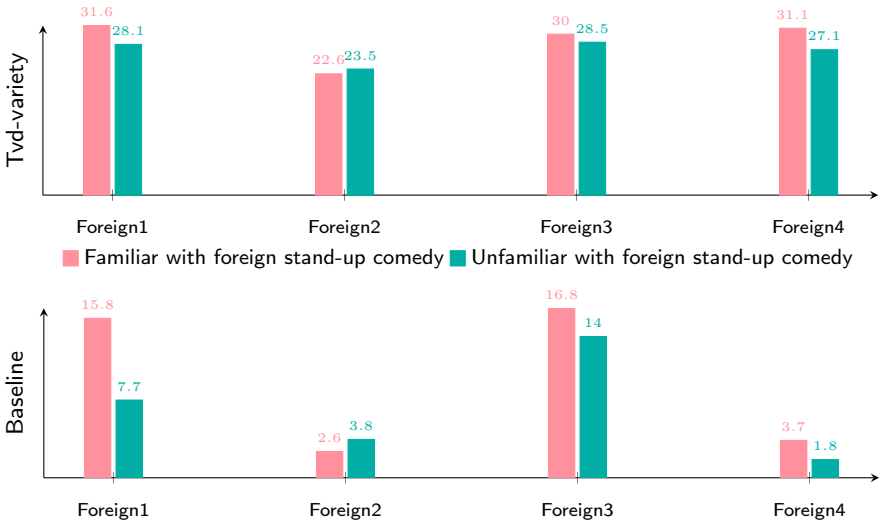
\begin{figure}[!h]
\definecolor{col1}{HTML}{ff919e}
\definecolor{col2}{HTML}{00aea5}
\centering
\begin{tikzpicture}[font=\small]
\sffamily
\begin{axis}[
  ybar,
  bar width=10pt,
  ylabel={Tvd-variety},
  ylabel near ticks,
  width = 5in, height = 1.5in,
  ymin=0,
  ytick=\empty,
  xtick=data,
  tick label style={font=\scriptsize},
  axis x line=bottom,
  axis y line=left,
  enlarge x limits=0.1,
  symbolic x coords={Foreign1,Foreign2,Foreign3,Foreign4},
  xticklabel style={anchor=base,yshift=-\baselineskip},
  nodes near coords={\pgfmathprintnumber\pgfplotspointmeta},
  %x tick label style={font=\tiny,},
  legend image code/.code={
        \draw [#1] (0cm,-0.1cm) rectangle (0.2cm,0.1cm); },
  legend style={
    font=\small,
    nodes={scale=1, transform shape},
    at={(0.5,-0.25)},
    anchor=north,
    draw=none,
    legend columns=-1},
]
  \addplot[fill=white,bar shift=-6pt,color=col1,font=\tiny] coordinates {
(Foreign1,31.6)
(Foreign2,22.6)
(Foreign3,30.0)
(Foreign4,31.1)
  };
  \addlegendentry{\footnotesize Familiar with foreign stand-up comedy}
  \addplot[fill=white,bar shift=6pt,color=col2,font=\tiny] coordinates {
(Foreign1,28.1)
(Foreign2,23.5)
(Foreign3,28.5)
(Foreign4,27.1)
  };
  \addlegendentry{\footnotesize Unfamiliar with foreign stand-up comedy}
\end{axis}
\end{tikzpicture}
\begin{tikzpicture}[font=\small]
\sffamily
\begin{axis}[
  ybar,
  bar width=10pt,
  ylabel={\textcolor{white}{|}Baseline\textcolor{white}{|}},
  ylabel near ticks,
  width = 5in, height = 1.5in,
  ymin=0,
  ytick=\empty,
  xtick=data,
  tick label style={font=\scriptsize},
  axis x line=bottom,
  axis y line=left,
  enlarge x limits=0.1,
  symbolic x coords={Foreign1,Foreign2,Foreign3,Foreign4},
  xticklabel style={anchor=base,yshift=-\baselineskip},
  nodes near coords={\pgfmathprintnumber\pgfplotspointmeta},
  %x tick label style={font=\tiny,},
  legend image code/.code={
        \draw [#1] (0cm,-0.1cm) rectangle (0.2cm,0.1cm); },
  legend style={
    font=\small,
    nodes={scale=1, transform shape},
    at={(0.5,-0.25)},
    anchor=north,
    draw=none,
    legend columns=-1},
]
  \addplot[fill=white,bar shift=-6pt,color=col1,font=\tiny] coordinates {
(Foreign1,15.8)
(Foreign2,2.6)
(Foreign3,16.8)
(Foreign4,3.7)
  };
  \addplot[fill=white,bar shift=6pt,color=col2,font=\tiny] coordinates {
(Foreign1,7.7)
(Foreign2,3.8)
(Foreign3,14.0)
(Foreign4,1.8)
  };
  
\end{axis}
\end{tikzpicture}
\captionof{figure}{\small {\bf Familiar with foreign stand-up comedy {\sl v.s.} unfamiliar:} There are four questions that compare two foreign stand-up comedians. The red bar represents the metrics of the group that contains people who are familiar with foreign stand-up comedy, while the green bar represents the metrics of the group who are unfamiliar with foreign stand-up comedy. The y-axis is 100 times of metrics.}
\label{talk:foreign}
\end{figure}

%% file: case_talkshow/crossquestion.tex
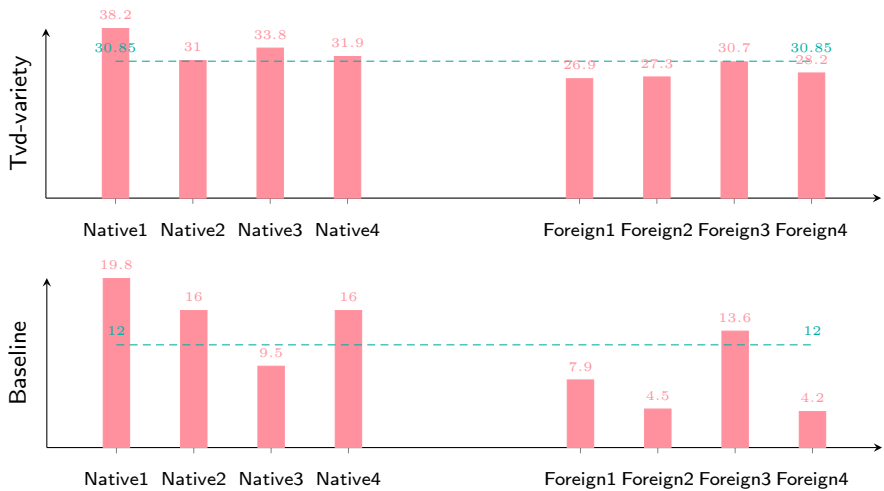
\begin{figure}[!h]
\definecolor{col1}{HTML}{ff919e}
\definecolor{col2}{HTML}{00aea5}
\centering
\begin{tikzpicture}[font=\small]
\sffamily
\begin{axis}[
  ybar,
  bar width=10pt,
  ylabel={Tvd-variety},
  ylabel near ticks,
  width = 5in, height = 1.5in,
  ymin=0,
  ytick=\empty,
  xtick=data,
  tick label style={font=\scriptsize},
  axis x line=bottom,
  axis y line=left,
  enlarge x limits=0.1,
  symbolic x coords={Native1,Native2,Native3,Native4,\textcolor{white}{|},\textcolor{white}{^},Foreign1,Foreign2,Foreign3,Foreign4},
  xticklabel style={anchor=base,yshift=-\baselineskip},
  nodes near coords={\pgfmathprintnumber\pgfplotspointmeta},
  %x tick label style={font=\tiny,},
  legend image code/.code={
        \draw [#1] (0cm,-0.1cm) rectangle (0.2cm,0.1cm); },
  legend style={
    font=\small,
    nodes={scale=1, transform shape},
    at={(0.5,-0.25)},
    anchor=north,
    draw=none,
    legend columns=-1},
]
  \addplot[fill=white,bar shift=0pt,color=col1,font=\tiny] coordinates {
(Native1,38.2)
(Foreign1,26.9)
(Native2,31.0)
(Foreign2,27.3)
(Native3,33.8)
(Foreign3,30.7)
(Native4,31.9)
(Foreign4,28.2)
  };
  \addplot[col2,sharp plot,densely dashed,font=\fontsize{0.0001}{0.0002}] coordinates
  {(Native1,30.85) (Foreign4,30.85)};
\end{axis}
\end{tikzpicture}

\begin{tikzpicture}[font=\small]
\sffamily
\begin{axis}[
  ybar,
  bar width=10pt,
  ylabel={\textcolor{white}{|}Baseline\textcolor{white}{|}},
  ylabel near ticks,
  width = 5in, height = 1.5in,
  ymin=0,
  ytick=\empty,
  xtick=data,
  tick label style={font=\scriptsize},
  axis x line=bottom,
  axis y line=left,
  enlarge x limits=0.1,
  symbolic x coords={Native1,Native2,Native3,Native4,\textcolor{white}{|},\textcolor{white}{^},Foreign1,Foreign2,Foreign3,Foreign4},
  xticklabel style={anchor=base,yshift=-\baselineskip},
  nodes near coords={\pgfmathprintnumber\pgfplotspointmeta},
  %x tick label style={font=\tiny,},
  legend image code/.code={
        \draw [#1] (0cm,-0.1cm) rectangle (0.2cm,0.1cm); },
  legend style={
    font=\small,
    nodes={scale=1, transform shape},
    at={(0.5,-0.25)},
    anchor=north,
    draw=none,
    legend columns=-1},
]
  
  \addplot[fill=white,bar shift=0pt,color=col1,font=\tiny] coordinates {
(Native1,19.8)
(Foreign1,7.9)
(Native2,16.0)
(Foreign2,4.5)
(Native3,9.5)
(Foreign3,13.6)
(Native4,16.0)
(Foreign4,4.2)
  };
  \addplot[col2,sharp plot,densely dashed,font=\fontsize{0.0001}{0.0002}] coordinates
  {(Native1,12) (Foreign4,12)};
\end{axis}
\end{tikzpicture}

\captionof{figure}{\small {\bf Easy and hard questions:} We divide the questions into two types. For type native (the left), these questions only compare native comedians and for type foreign (the right), these questions only compare foreign comedians. We pick a group of respondents who are familiar with native stand-up comedy but not familiar with foreign stand-up comedy such that questions of type native are easier for them, compared to questions of type foreign. We compute both the Tvd-variety and baseline score of their feedback for all questions. The green lines show that Tvd-variety successfully separates easy and hard questions (easy questions are above the line and hard questions are below the line) while baseline does not. 
}
\label{talk:crossq}
\end{figure}

%% file: case_athletes/sex.tex
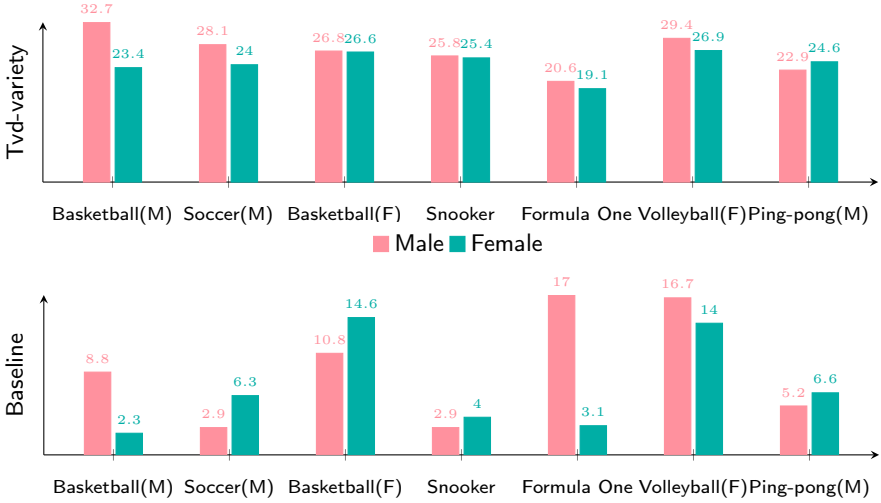
\begin{figure}[!h]
\definecolor{col1}{HTML}{ff919e}
\definecolor{col2}{HTML}{00aea5}
\centering
\begin{tikzpicture}[font=\small]
\sffamily
\begin{axis}[
  ybar,
  bar width=10pt,
  ylabel={Tvd-variety},
  ylabel near ticks,
  width = 5in, height = 1.45in,
  ymin=0,
  ytick=\empty,
  xtick=data,
  tick label style={font=\scriptsize},
  axis x line=bottom,
  axis y line=left,
  enlarge x limits=0.1,
  symbolic x coords={Basketball(M),Soccer(M),Basketball(F),Snooker,Formula One,Volleyball(F),Ping-pong(M)},
  xticklabel style={anchor=base,yshift=-\baselineskip},
  nodes near coords={\pgfmathprintnumber\pgfplotspointmeta},
  %x tick label style={font=\tiny,},
  legend image code/.code={
        \draw [#1] (0cm,-0.1cm) rectangle (0.2cm,0.1cm); },
  legend style={
    font=\small,
    nodes={scale=1, transform shape},
    at={(0.5,-0.25)},
    anchor=north,
    draw=none,
    legend columns=-1},
]
  \addplot[fill=white,bar shift=-6pt,color=col1,font=\tiny] coordinates {
(Basketball(M),32.7)
(Soccer(M),28.1)
(Basketball(F),26.8)
(Snooker,25.8)
(Formula One,20.6)
(Volleyball(F),29.4)
(Ping-pong(M),22.9)
  };
  \addlegendentry{Male}
  \addplot[fill=white,bar shift=6pt,color=col2,font=\tiny] coordinates {
(Basketball(M),23.4)
(Soccer(M),24.0)
(Basketball(F),26.6)
(Snooker,25.4)
(Formula One,19.1)
(Volleyball(F),26.9)
(Ping-pong(M),24.6)
  };
  \addlegendentry{Female}
\end{axis}
\end{tikzpicture}
\begin{tikzpicture}[font=\small]
\sffamily
\begin{axis}[
  ybar,
  bar width=10pt,
  ylabel={{\textcolor{white}{|}Baseline\textcolor{white}{|}}},
  ylabel near ticks,
  width = 5in, height = 1.45in,
  ymin=0,
  ytick=\empty,
  xtick=data,
  tick label style={font=\scriptsize},
  axis x line=bottom,
  axis y line=left,
  enlarge x limits=0.1,
  symbolic x coords={Basketball(M),Soccer(M),Basketball(F),Snooker,Formula One,Volleyball(F),Ping-pong(M)},
  xticklabel style={anchor=base,yshift=-\baselineskip},
  nodes near coords={\pgfmathprintnumber\pgfplotspointmeta},
  %x tick label style={font=\tiny,},
  legend image code/.code={
        \draw [#1] (0cm,-0.1cm) rectangle (0.2cm,0.1cm); },
  legend style={
    font=\small,
    nodes={scale=1, transform shape},
    at={(0.5,-0.25)},
    anchor=north,
    draw=none,
    legend columns=-1},
]
  \addplot[fill=white,bar shift=-6pt,color=col1,font=\tiny] coordinates {
(Basketball(M),8.8)
(Soccer(M),2.9)
(Basketball(F),10.8)
(Snooker,2.9)
(Formula One,17.0)
(Volleyball(F),16.7)
(Ping-pong(M),5.2)
  };
  
  \addplot[fill=white,bar shift=6pt,color=col2,font=\tiny] coordinates {
(Basketball(M),2.3)
(Soccer(M),6.3)
(Basketball(F),14.6)
(Snooker,4.0)
(Formula One,3.1)
(Volleyball(F),14.0)
(Ping-pong(M),6.6)
  };
  
\end{axis}
\end{tikzpicture}
\captionof{figure}{\small {\bf Athletes study: male {\sl v.s.} female:} We also compare the male group and female group of respondents in the study for comparisons between athletes. }
\label{ath:sex}
\end{figure}

%% file: case_talkshow/sex.tex
\begin{figure}[!h]
\definecolor{col1}{HTML}{ff919e}
\definecolor{col2}{HTML}{00aea5}
\centering
\begin{tikzpicture}[font=\small]
\sffamily
\begin{axis}[
  ybar,
  bar width=10pt,
  ylabel={Tvd-variety},
  ylabel near ticks,
  width = 5in, height = 1.45in,
  ymin=0,
  ytick=\empty,
  xtick=data,
  tick label style={font=\scriptsize},
  axis x line=bottom,
  axis y line=left,
  enlarge x limits=0.1,
  symbolic x coords={Native1,Native2,Native3,Native4,Foreign1,Foreign2,Foreign3,Foreign4},
  xticklabel style={anchor=base,yshift=-\baselineskip},
  nodes near coords={\pgfmathprintnumber\pgfplotspointmeta},
  %x tick label style={font=\tiny,},
  legend image code/.code={
        \draw [#1] (0cm,-0.1cm) rectangle (0.2cm,0.1cm); },
  legend style={
    font=\small,
    nodes={scale=1, transform shape},
    at={(0.5,-0.25)},
    anchor=north,
    draw=none,
    legend columns=-1},
]
  \addplot[fill=white,bar shift=-6pt,color=col1,font=\tiny] coordinates {
(Native1,40.1)
(Foreign1,31.3)
(Native2,30.5)
(Foreign2,28.6)
(Native3,32.8)
(Foreign3,30.9)
(Native4,35.5)
(Foreign4,32.1)
  };
  \addlegendentry{Male}
  \addplot[fill=white,bar shift=6pt,color=col2,font=\tiny] coordinates {
(Native1,36.5)
(Foreign1,27.6)
(Native2,27.3)
(Foreign2,20.3)
(Native3,31.9)
(Foreign3,27.3)
(Native4,29.5)
(Foreign4,25.7)
  };
  \addlegendentry{Female}
\end{axis}
\end{tikzpicture}
\begin{tikzpicture}[font=\small]
\sffamily
\begin{axis}[
  ybar,
  bar width=10pt,
  ylabel={\textcolor{white}{|}Baseline\textcolor{white}{|}},
  ylabel near ticks,
  width = 5in, height = 1.45in,
  ymin=0,
  ytick=\empty,
  xtick=data,
  tick label style={font=\scriptsize},
  axis x line=bottom,
  axis y line=left,
  enlarge x limits=0.1,
  symbolic x coords={Native1,Native2,Native3,Native4,Foreign1,Foreign2,Foreign3,Foreign4},
  xticklabel style={anchor=base,yshift=-\baselineskip},
  nodes near coords={\pgfmathprintnumber\pgfplotspointmeta},
  %x tick label style={font=\tiny,},
  legend image code/.code={
        \draw [#1] (0cm,-0.1cm) rectangle (0.2cm,0.1cm); },
  legend style={
    font=\small,
    nodes={scale=1, transform shape},
    at={(0.5,-0.25)},
    anchor=north,
    draw=none,
    legend columns=-1},
]
  \addplot[fill=white,bar shift=-6pt,color=col1,font=\tiny] coordinates {
(Native1,13.7)
(Foreign1,0.4)
(Native2,14.5)
(Foreign2,9.2)
(Native3,3.8)
(Foreign3,9.5)
(Native4,15.3)
(Foreign4,3.4)
  };
  \addplot[fill=white,bar shift=6pt,color=col2,font=\tiny] coordinates {
(Native1,13.0)
(Foreign1,1.4)
(Native2,16.2)
(Foreign2,0.5)
(Native3,16.5)
(Foreign3,18.6)
(Native4,17.6)
(Foreign4,1.6)
  };
  
\end{axis}
\end{tikzpicture}
\captionof{figure}{\small {\bf Comedians study: male {\sl v.s.} female:} We also compare the male group and female group of respondents in the study for comparisons between stand-up comedians. }
\label{talk:sex}
\end{figure}
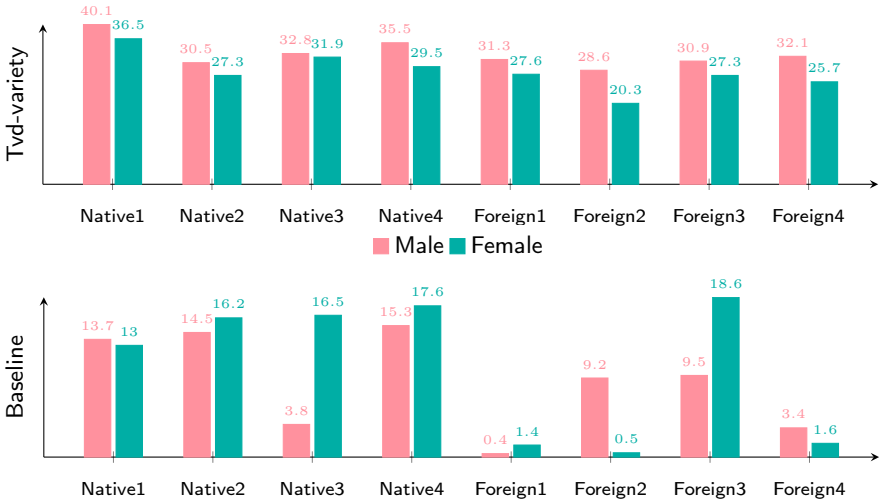

%% file: numerical_fig/pearson-variety.tex
\begin{figure}[!h]
    \begin{minipage}[b]{0.48\textwidth}
        \centering
        \input{numerical_fig/balancebeta1-pearson}
        \captionof*{figure}{\small (a) {\bf Uniform-1}}
    \end{minipage}
    \hfill
    \begin{minipage}[b]{0.48\textwidth}
        \centering
        \input{numerical_fig/unbalancebeta1-pearson}
        \captionof*{figure}{\small (b) {\bf Non-Uniform-1}}
    \end{minipage}
    \vfill
    \begin{minipage}[b]{0.48\textwidth}
        \centering
        \input{numerical_fig/balancebeta2-pearson}
        \captionof*{figure}{\small (c) {\bf Uniform-2}}
    \end{minipage}
    \hfill
    \begin{minipage}[b]{0.48\textwidth}
        \centering
        \input{numerical_fig/unbalancebeta2-pearson}
        \captionof*{figure}{\small (d) {\bf Non-Uniform-2}}
    \end{minipage}
    
    \captionof{figure}{\small {\bf Pearson-variety {\sl v.s.} \bf ratio of non-experts:} We adopt the same setting as Figure~\ref{fig:numfig} and observe the similar results for Pearson-variety.}
    \label{fig:numfig-pearson}
\end{figure}

%% file: numerical_fig/balancebeta1-pearson.tex
\begin{tikzpicture}[scale=.5]
\sffamily
\begin{axis}[
title style={align=center,yshift=-.0in},
legend style={font=\large,
 nodes={scale=1, transform shape},
 at={(.7,.9)},
 anchor=north west,
 draw=none},
legend cell align={left},
width = 4.5in, height = 3in,
ylabel near ticks,
ytick={0, .1, .2, .3, .4, .5, .6, .7, .8},
ylabel = {Pearson-variety},
xlabel near ticks,
every tick label/.append style={font=\normalsize},
xmin=-0.05,xmax=1.05,ymin=0,ymax=.8,
xtick={0, .1, .2, .3, .4, .5, .6, .7, .8, .9, 1},
xlabel={ratio of uninformative participants},
xlabel style = {font=\Large},
ylabel style = {font=\Large},
yticklabel style={
  /pgf/number format/fixed,
  /pgf/number format/precision=5
},
%scaled y ticks=false
%every axis plot/.append style={thick}
]

\addplot[solid, mark=., mark options={scale=.8}, gray, style=thick,]
table[]{
0 0.248505
0.01 0.24347
0.02 0.238545
0.03 0.233713
0.04 0.228965
0.05 0.224292
0.06 0.21969
0.07 0.215155
0.08 0.210683
0.09 0.20627
0.1 0.201916
0.11 0.197616
0.12 0.193371
0.13 0.189179
0.14 0.185037
0.15 0.180945
0.16 0.176903
0.17 0.172909
0.18 0.168962
0.19 0.165063
0.2 0.161209
0.21 0.157401
0.22 0.153639
0.23 0.149922
0.24 0.146249
0.25 0.14262
0.26 0.139036
0.27 0.135496
0.28 0.131999
0.29 0.128545
0.3 0.125135
0.31 0.121769
0.32 0.118445
0.33 0.115164
0.34 0.111927
0.35 0.108733
0.36 0.105581
0.37 0.102473
0.38 0.0994078
0.39 0.0963856
0.4 0.0934065
0.41 0.0904706
0.42 0.087578
0.43 0.0847287
0.44 0.0819228
0.45 0.0791604
0.46 0.0764417
0.47 0.0737667
0.48 0.0711356
0.49 0.0685484
0.5 0.0660054
0.51 0.0635068
0.52 0.0610525
0.53 0.0586429
0.54 0.056278
0.55 0.0539582
0.56 0.0516834
0.57 0.049454
0.58 0.0472701
0.59 0.0451319
0.6 0.0430397
0.61 0.0409936
0.62 0.0389939
0.63 0.0370407
0.64 0.0351344
0.65 0.0332751
0.66 0.0314631
0.67 0.0296986
0.68 0.0279819
0.69 0.0263132
0.7 0.0246928
0.71 0.023121
0.72 0.021598
0.73 0.020124
0.74 0.0186994
0.75 0.0173245
0.76 0.0159994
0.77 0.0147246
0.78 0.0135004
0.79 0.0123269
0.8 0.0112045
0.81 0.0101336
0.82 0.00911439
0.83 0.00814725
0.84 0.00723249
0.85 0.00637043
0.86 0.00556141
0.87 0.00480576
0.88 0.00410381
0.89 0.00345593
0.9 0.00286244
0.91 0.00232372
0.92 0.00184011
0.93 0.00141198
0.94 0.0010397
0.95 0.000723642
0.96 0.000464178
0.97 0.000261693
0.98 0.000116573
0.99 0.0000292099
1. 0
};

\addlegendentry{theoretical}

\addplot[solid, mark=o, mark options={scale=.8}, red]
plot [error bars/.cd, y dir = both, y explicit]
table[y error index=2]{
0.0030 0.2519 0.0224
0.1030 0.2082 0.0211
0.2030 0.1675 0.0197
0.3030 0.1325 0.0186
0.4030 0.1014 0.0165
0.5030 0.0744 0.0141
0.6030 0.0527 0.0131
0.7030 0.0349 0.0103
0.8030 0.0223 0.0079
0.9030 0.0141 0.0058
1.0030 0.0109 0.0046
};
\addlegendentry{$n=1000$}

\addplot[solid, mark=o, mark options={scale=.8}, orange]
plot [error bars/.cd, y dir = both, y explicit]
table[y error index=2]{
0.0010 0.2583 0.0313
0.1010 0.2117 0.0298
0.2010 0.1764 0.0278
0.3010 0.1412 0.0271
0.4010 0.1119 0.0252
0.5010 0.0851 0.0213
0.6010 0.0629 0.0181
0.7010 0.0448 0.0144
0.8010 0.0319 0.0123
0.9010 0.0242 0.0099
1.0010 0.0224 0.0091
};
\addlegendentry{$n=500$}

\addplot[solid, mark=o, mark options={scale=.8}, green]
plot [error bars/.cd, y dir = both, y explicit]
table[y error index=2]{
-0.0010 0.2745 0.0491
0.0990 0.2351 0.0459
0.1990 0.1966 0.0459
0.2990 0.1674 0.0448
0.3990 0.1365 0.0372
0.4990 0.1097 0.0353
0.5990 0.0921 0.0327
0.6990 0.0759 0.0290
0.7990 0.0628 0.0240
0.8990 0.0550 0.0231
0.9990 0.0519 0.0209
};
\addlegendentry{$n=200$}

\addplot[solid, mark=o, mark options={scale=.8}, cyan]
plot [error bars/.cd, y dir = both, y explicit]
table[y error index=2]{
-0.0030 0.3060 0.0687
0.0970 0.2644 0.0699
0.1970 0.2339 0.0667
0.2970 0.2040 0.0644
0.3970 0.1759 0.0567
0.4970 0.1535 0.0548
0.5970 0.1376 0.0516
0.6970 0.1190 0.0463
0.7970 0.1080 0.0413
0.8970 0.1036 0.0436
0.9970 0.1014 0.0414
};
\addlegendentry{$n=100$}

\end{axis}
\end{tikzpicture}

%% file: numerical_fig/unbalancebeta1-pearson.tex
\begin{tikzpicture}[scale=.5]
\sffamily
\begin{axis}[
title style={align=center,yshift=-.0in},
legend style={font=\large,
 nodes={scale=1, transform shape},
 at={(.7,.9)},
 anchor=north west,
 draw=none},
legend cell align={left},
width = 4.5in, height = 3in,
ylabel near ticks,
ytick={0, .1, .2, .3, .4, .5, .6, .7, .8},
ylabel = {Pearson-variety},
xlabel near ticks,
every tick label/.append style={font=\normalsize},
xmin=-0.05,xmax=1.05,ymin=0,ymax=.8,
xtick={0, .1, .2, .3, .4, .5, .6, .7, .8, .9, 1},
xlabel={ratio of uninformative participants},
xlabel style = {font=\Large},
ylabel style = {font=\Large},
yticklabel style={
  /pgf/number format/fixed,
  /pgf/number format/precision=5
},
%scaled y ticks=false
%every axis plot/.append style={thick}
]

\addplot[solid, mark=., mark options={scale=.8}, gray, style=thick,]
table[]{
0 0.572854
0.01 0.561116
0.02 0.549555
0.03 0.538155
0.04 0.526908
0.05 0.515808
0.06 0.504851
0.07 0.494033
0.08 0.483352
0.09 0.472805
0.1 0.462391
0.11 0.452106
0.12 0.441949
0.13 0.43192
0.14 0.422017
0.15 0.412238
0.16 0.402582
0.17 0.39305
0.18 0.383638
0.19 0.374348
0.2 0.365177
0.21 0.356126
0.22 0.347194
0.23 0.338379
0.24 0.329682
0.25 0.321102
0.26 0.312639
0.27 0.304292
0.28 0.29606
0.29 0.287944
0.3 0.279943
0.31 0.272056
0.32 0.264283
0.33 0.256625
0.34 0.24908
0.35 0.241649
0.36 0.234331
0.37 0.227126
0.38 0.220035
0.39 0.213055
0.4 0.206189
0.41 0.199435
0.42 0.192793
0.43 0.186263
0.44 0.179845
0.45 0.17354
0.46 0.167346
0.47 0.161264
0.48 0.155293
0.49 0.149435
0.5 0.143688
0.51 0.138053
0.52 0.132529
0.53 0.127117
0.54 0.121817
0.55 0.116628
0.56 0.11155
0.57 0.106585
0.58 0.101731
0.59 0.0969887
0.6 0.0923582
0.61 0.0878395
0.62 0.0834326
0.63 0.0791377
0.64 0.0749548
0.65 0.0708839
0.66 0.0669252
0.67 0.0630786
0.68 0.0593445
0.69 0.0557227
0.7 0.0522135
0.71 0.0488169
0.72 0.0455331
0.73 0.0423622
0.74 0.0393043
0.75 0.0363595
0.76 0.033528
0.77 0.03081
0.78 0.0282055
0.79 0.0257147
0.8 0.0233379
0.81 0.0210751
0.82 0.0189265
0.83 0.0168922
0.84 0.0149726
0.85 0.0131677
0.86 0.0114777
0.87 0.00990288
0.88 0.00844334
0.89 0.00709932
0.9 0.00587102
0.91 0.00475864
0.92 0.00376239
0.93 0.0028825
0.94 0.00211917
0.95 0.00147264
0.96 0.000943134
0.97 0.000530877
0.98 0.000236108
0.99 0.0000590682
1. 0
};

\addlegendentry{theoretical}

\addplot[solid, mark=o, mark options={scale=.8}, red]
plot [error bars/.cd, y dir = both, y explicit]
table[y error index=2]{
0.0030 0.5670 0.0244
0.1030 0.4601 0.0259
0.2030 0.3648 0.0257
0.3030 0.2831 0.0241
0.4030 0.2111 0.0225
0.5030 0.1508 0.0197
0.6030 0.0998 0.0169
0.7030 0.0612 0.0132
0.8030 0.0333 0.0097
0.9030 0.0167 0.0067
1.0030 0.0110 0.0047
};
\addlegendentry{$n=1000$}

\addplot[solid, mark=o, mark options={scale=.8}, orange]
plot [error bars/.cd, y dir = both, y explicit]
table[y error index=2]{
0.0010 0.5678 0.0363
0.1010 0.4656 0.0364
0.2010 0.3711 0.0379
0.3010 0.2915 0.0347
0.4010 0.2200 0.0311
0.5010 0.1598 0.0281
0.6010 0.1114 0.0261
0.7010 0.0720 0.0200
0.8010 0.0442 0.0151
0.9010 0.0276 0.0112
1.0010 0.0221 0.0090
};
\addlegendentry{$n=500$}

\addplot[solid, mark=o, mark options={scale=.8}, green]
plot [error bars/.cd, y dir = both, y explicit]
table[y error index=2]{
-0.0010 0.5792 0.0558
0.0990 0.4763 0.0580
0.1990 0.3890 0.0571
0.2990 0.3060 0.0537
0.3990 0.2448 0.0510
0.4990 0.1844 0.0461
0.5990 0.1375 0.0400
0.6990 0.0990 0.0334
0.7990 0.0736 0.0284
0.8990 0.0578 0.0242
0.9990 0.0527 0.0217
};
\addlegendentry{$n=200$}

\addplot[solid, mark=o, mark options={scale=.8}, cyan]
plot [error bars/.cd, y dir = both, y explicit]
table[y error index=2]{
-0.0030 0.5944 0.0761
0.0970 0.4964 0.0752
0.1970 0.4177 0.0792
0.2970 0.3423 0.0774
0.3970 0.2816 0.0750
0.4970 0.2191 0.0680
0.5970 0.1788 0.0584
0.6970 0.1409 0.0504
0.7970 0.1192 0.0481
0.8970 0.1053 0.0431
0.9970 0.1011 0.0432
};
\addlegendentry{$n=100$}

\end{axis}
\end{tikzpicture}

%% file: numerical_fig/balancebeta2-pearson.tex
\begin{tikzpicture}[scale=.5]
\sffamily
\begin{axis}[
title style={align=center,yshift=-.0in},
legend style={font=\large,
 nodes={scale=1, transform shape},
 at={(.7,.9)},
 anchor=north west,
 draw=none},
legend cell align={left},
width = 4.5in, height = 3in,
ylabel near ticks,
ytick={0, .1, .2, .3, .4, .5, .6, .7, .8},
ylabel = {Pearson-variety},
xlabel near ticks,
every tick label/.append style={font=\normalsize},
xmin=-0.05,xmax=1.05,ymin=0,ymax=.8,
xtick={0, .1, .2, .3, .4, .5, .6, .7, .8, .9, 1},
xlabel={ratio of uninformative participants},
xlabel style = {font=\Large},
ylabel style = {font=\Large},
yticklabel style={
  /pgf/number format/fixed,
  /pgf/number format/precision=5
},
%scaled y ticks=false
%every axis plot/.append style={thick}
]

\addplot[solid, mark=., mark options={scale=.8}, gray, style=thick,]
table[]{
0 0.145076
0.01 0.141922
0.02 0.138859
0.03 0.135867
0.04 0.132936
0.05 0.130059
0.06 0.127231
0.07 0.124449
0.08 0.121709
0.09 0.11901
0.1 0.11635
0.11 0.113728
0.12 0.111141
0.13 0.108591
0.14 0.106075
0.15 0.103593
0.16 0.101144
0.17 0.0987278
0.18 0.0963441
0.19 0.0939924
0.2 0.0916723
0.21 0.0893835
0.22 0.0871257
0.23 0.0848986
0.24 0.0827021
0.25 0.0805358
0.26 0.0783998
0.27 0.0762936
0.28 0.0742173
0.29 0.0721707
0.3 0.0701536
0.31 0.0681659
0.32 0.0662076
0.33 0.0642785
0.34 0.0623785
0.35 0.0605076
0.36 0.0586657
0.37 0.0568526
0.38 0.0550684
0.39 0.053313
0.4 0.0515864
0.41 0.0498884
0.42 0.048219
0.43 0.0465782
0.44 0.044966
0.45 0.0433822
0.46 0.041827
0.47 0.0403002
0.48 0.0388018
0.49 0.0373318
0.5 0.0358902
0.51 0.034477
0.52 0.033092
0.53 0.0317354
0.54 0.0304071
0.55 0.0291071
0.56 0.0278354
0.57 0.026592
0.58 0.0253768
0.59 0.0241899
0.6 0.0230312
0.61 0.0219008
0.62 0.0207987
0.63 0.0197248
0.64 0.0186792
0.65 0.0176619
0.66 0.0166728
0.67 0.0157121
0.68 0.0147796
0.69 0.0138754
0.7 0.0129996
0.71 0.012152
0.72 0.0113328
0.73 0.010542
0.74 0.00977952
0.75 0.00904545
0.76 0.00833978
0.77 0.00766255
0.78 0.00701378
0.79 0.00639348
0.8 0.00580168
0.81 0.00523841
0.82 0.00470369
0.83 0.00419755
0.84 0.00372003
0.85 0.00327114
0.86 0.00285094
0.87 0.00245944
0.88 0.00209668
0.89 0.00176271
0.9 0.00145755
0.91 0.00118124
0.92 0.000933829
0.93 0.000715354
0.94 0.000525857
0.95 0.000365384
0.96 0.000233979
0.97 0.000131689
0.98 0.0000585629
0.99 0.0000146494
1. 0
};

\addlegendentry{theoretical}

\addplot[solid, mark=o, mark options={scale=.8}, red]
plot [error bars/.cd, y dir = both, y explicit]
table[y error index=2]{
0.0030 0.1468 0.0184
0.1030 0.1205 0.0178
0.2030 0.0980 0.0163
0.3030 0.0773 0.0148
0.4030 0.0596 0.0134
0.5030 0.0447 0.0116
0.6030 0.0324 0.0096
0.7030 0.0230 0.0078
0.8030 0.0168 0.0064
0.9030 0.0124 0.0052
1.0030 0.0109 0.0046
};
\addlegendentry{$n=1000$}

\addplot[solid, mark=o, mark options={scale=.8}, orange]
plot [error bars/.cd, y dir = both, y explicit]
table[y error index=2]{
0.0010 0.1529 0.0257
0.1010 0.1275 0.0247
0.2010 0.1059 0.0227
0.3010 0.0868 0.0209
0.4010 0.0678 0.0188
0.5010 0.0549 0.0174
0.6010 0.0429 0.0146
0.7010 0.0341 0.0129
0.8010 0.0272 0.0106
0.9010 0.0225 0.0093
1.0010 0.0214 0.0088
};
\addlegendentry{$n=500$}

\addplot[solid, mark=o, mark options={scale=.8}, green]
plot [error bars/.cd, y dir = both, y explicit]
table[y error index=2]{
-0.0010 0.1735 0.0427
0.0990 0.1508 0.0414
0.1990 0.1311 0.0388
0.2990 0.1129 0.0366
0.3990 0.0966 0.0320
0.4990 0.0825 0.0299
0.5990 0.0724 0.0283
0.6990 0.0626 0.0258
0.7990 0.0573 0.0236
0.8990 0.0542 0.0229
0.9990 0.0519 0.0207
};
\addlegendentry{$n=200$}

\addplot[solid, mark=o, mark options={scale=.8}, cyan]
plot [error bars/.cd, y dir = both, y explicit]
table[y error index=2]{
-0.0030 0.2071 0.0620
0.0970 0.1840 0.0626
0.1970 0.1670 0.0580
0.2970 0.1502 0.0535
0.3970 0.1369 0.0519
0.4970 0.1276 0.0502
0.5970 0.1166 0.0459
0.6970 0.1107 0.0441
0.7970 0.1056 0.0436
0.8970 0.1005 0.0397
0.9970 0.0992 0.0397
};
\addlegendentry{$n=100$}

\end{axis}
\end{tikzpicture}

%% file: numerical_fig/unbalancebeta2-pearson.tex
\begin{tikzpicture}[scale=.5]
\sffamily
\begin{axis}[
title style={align=center,yshift=-.0in},
legend style={font=\large,
 nodes={scale=1, transform shape},
 at={(.7,.9)},
 anchor=north west,
 draw=none},
legend cell align={left},
width = 4.5in, height = 3in,
ylabel near ticks,
ytick={0, .1, .2, .3, .4, .5, .6, .7, .8},
ylabel = {Pearson-variety},
xlabel near ticks,
every tick label/.append style={font=\normalsize},
xmin=-0.05,xmax=1.05,ymin=0,ymax=.8,
xtick={0, .1, .2, .3, .4, .5, .6, .7, .8, .9, 1},
xlabel={ratio of uninformative participants},
xlabel style = {font=\Large},
ylabel style = {font=\Large},
yticklabel style={
  /pgf/number format/fixed,
  /pgf/number format/precision=5
},
%scaled y ticks=false
%every axis plot/.append style={thick}
]

\addplot[solid, mark=., mark options={scale=.8}, gray, style=thick,]
table[]{
0 0.248505
0.01 0.24347
0.02 0.238545
0.03 0.233713
0.04 0.228965
0.05 0.224292
0.06 0.21969
0.07 0.215155
0.08 0.210683
0.09 0.20627
0.1 0.201916
0.11 0.197616
0.12 0.193371
0.13 0.189179
0.14 0.185037
0.15 0.180945
0.16 0.176903
0.17 0.172909
0.18 0.168962
0.19 0.165063
0.2 0.161209
0.21 0.157401
0.22 0.153639
0.23 0.149922
0.24 0.146249
0.25 0.14262
0.26 0.139036
0.27 0.135496
0.28 0.131999
0.29 0.128545
0.3 0.125135
0.31 0.121769
0.32 0.118445
0.33 0.115164
0.34 0.111927
0.35 0.108733
0.36 0.105581
0.37 0.102473
0.38 0.0994078
0.39 0.0963856
0.4 0.0934065
0.41 0.0904706
0.42 0.087578
0.43 0.0847287
0.44 0.0819228
0.45 0.0791604
0.46 0.0764417
0.47 0.0737667
0.48 0.0711356
0.49 0.0685484
0.5 0.0660054
0.51 0.0635068
0.52 0.0610525
0.53 0.0586429
0.54 0.056278
0.55 0.0539582
0.56 0.0516834
0.57 0.049454
0.58 0.0472701
0.59 0.0451319
0.6 0.0430397
0.61 0.0409936
0.62 0.0389939
0.63 0.0370407
0.64 0.0351344
0.65 0.0332751
0.66 0.0314631
0.67 0.0296986
0.68 0.0279819
0.69 0.0263132
0.7 0.0246928
0.71 0.023121
0.72 0.021598
0.73 0.020124
0.74 0.0186994
0.75 0.0173245
0.76 0.0159994
0.77 0.0147246
0.78 0.0135004
0.79 0.0123269
0.8 0.0112045
0.81 0.0101336
0.82 0.00911439
0.83 0.00814725
0.84 0.00723249
0.85 0.00637043
0.86 0.00556141
0.87 0.00480576
0.88 0.00410381
0.89 0.00345593
0.9 0.00286244
0.91 0.00232372
0.92 0.00184011
0.93 0.00141198
0.94 0.0010397
0.95 0.000723642
0.96 0.000464178
0.97 0.000261693
0.98 0.000116573
0.99 0.0000292099
1. 0
};

\addlegendentry{theoretical}

\addplot[solid, mark=o, mark options={scale=.8}, red]
plot [error bars/.cd, y dir = both, y explicit]
table[y error index=2]{
0.0030 0.2519 0.0224
0.1030 0.2082 0.0211
0.2030 0.1675 0.0197
0.3030 0.1325 0.0186
0.4030 0.1014 0.0165
0.5030 0.0744 0.0141
0.6030 0.0527 0.0131
0.7030 0.0349 0.0103
0.8030 0.0223 0.0079
0.9030 0.0141 0.0058
1.0030 0.0109 0.0046
};
\addlegendentry{$n=1000$}

\addplot[solid, mark=o, mark options={scale=.8}, orange]
plot [error bars/.cd, y dir = both, y explicit]
table[y error index=2]{
0.0010 0.2583 0.0313
0.1010 0.2117 0.0298
0.2010 0.1764 0.0278
0.3010 0.1412 0.0271
0.4010 0.1119 0.0252
0.5010 0.0851 0.0213
0.6010 0.0629 0.0181
0.7010 0.0448 0.0144
0.8010 0.0319 0.0123
0.9010 0.0242 0.0099
1.0010 0.0224 0.0091
};
\addlegendentry{$n=500$}

\addplot[solid, mark=o, mark options={scale=.8}, green]
plot [error bars/.cd, y dir = both, y explicit]
table[y error index=2]{
-0.0010 0.2745 0.0491
0.0990 0.2351 0.0459
0.1990 0.1966 0.0459
0.2990 0.1674 0.0448
0.3990 0.1365 0.0372
0.4990 0.1097 0.0353
0.5990 0.0921 0.0327
0.6990 0.0759 0.0290
0.7990 0.0628 0.0240
0.8990 0.0550 0.0231
0.9990 0.0519 0.0209
};
\addlegendentry{$n=200$}

\addplot[solid, mark=o, mark options={scale=.8}, cyan]
plot [error bars/.cd, y dir = both, y explicit]
table[y error index=2]{
-0.0030 0.3060 0.0687
0.0970 0.2644 0.0699
0.1970 0.2339 0.0667
0.2970 0.2040 0.0644
0.3970 0.1759 0.0567
0.4970 0.1535 0.0548
0.5970 0.1376 0.0516
0.6970 0.1190 0.0463
0.7970 0.1080 0.0413
0.8970 0.1036 0.0436
0.9970 0.1014 0.0414
};
\addlegendentry{$n=100$}

\end{axis}
\end{tikzpicture}

%% file: numerical_fig/hellinger-variety.tex
\begin{figure}[!h]
    \begin{minipage}[b]{0.48\textwidth}
        \centering
        \input{numerical_fig/balancebeta1-hellinger}
        \captionof*{figure}{\small (a) {\bf Uniform-1}}
    \end{minipage}
    \hfill
    \begin{minipage}[b]{0.48\textwidth}
        \centering
        \input{numerical_fig/unbalancebeta1-hellinger}
        \captionof*{figure}{\small (b) {\bf Non-Uniform-1}}
    \end{minipage}
    \vfill
    \begin{minipage}[b]{0.48\textwidth}
        \centering
        \input{numerical_fig/balancebeta2-hellinger}
        \captionof*{figure}{\small (c) {\bf Uniform-2}}
    \end{minipage}
    \hfill
    \begin{minipage}[b]{0.48\textwidth}
        \centering
        \input{numerical_fig/unbalancebeta2-hellinger}
        \captionof*{figure}{\small (d) {\bf Non-Uniform-2}}
    \end{minipage}
    
    \captionof{figure}{\small {\bf Hellinger-variety {\sl v.s.} \bf ratio of non-experts:} We adopt the same setting as Figure~\ref{fig:numfig} and observe the similar results for Hellinger-variety.}
    \label{fig:numfig-hellinger}
\end{figure}

%% file: numerical_fig/balancebeta1-hellinger.tex
\begin{tikzpicture}[scale=.5]
\sffamily
\begin{axis}[
title style={align=center,yshift=-.0in},
legend style={font=\large,
 nodes={scale=1, transform shape},
 at={(.7,.9)},
 anchor=north west,
 draw=none},
legend cell align={left},
width = 4.5in, height = 3in,
ylabel near ticks,
ytick={0, .02, .04, .06, .08, .10, .12, .14},
ylabel = {Hellinger-variety},
xlabel near ticks,
every tick label/.append style={font=\normalsize},
xmin=-0.05,xmax=1.05,ymin=0,ymax=.14,
xtick={0, .1, .2, .3, .4, .5, .6, .7, .8, .9, 1},
xlabel={ratio of uninformative participants},
xlabel style = {font=\Large},
ylabel style = {font=\Large},
yticklabel style={
  /pgf/number format/fixed,
  /pgf/number format/precision=5
},
%scaled y ticks=false
%every axis plot/.append style={thick}
]

\addplot[solid, mark=., mark options={scale=.8}, gray, style=thick,]
table[]{
0 0.100891
0.01 0.0951316
0.02 0.090993
0.03 0.0874232
0.04 0.0842105
0.05 0.0812563
0.06 0.0785041
0.07 0.0759169
0.08 0.0734687
0.09 0.0711405
0.1 0.0689176
0.11 0.0667886
0.12 0.0647439
0.13 0.062776
0.14 0.0608783
0.15 0.0590454
0.16 0.0572725
0.17 0.0555557
0.18 0.0538913
0.19 0.0522762
0.2 0.0507075
0.21 0.0491829
0.22 0.0477001
0.23 0.0462569
0.24 0.0448518
0.25 0.0434829
0.26 0.0421487
0.27 0.0408479
0.28 0.0395792
0.29 0.0383415
0.3 0.0371335
0.31 0.0359544
0.32 0.0348031
0.33 0.0336789
0.34 0.0325809
0.35 0.0315084
0.36 0.0304606
0.37 0.0294369
0.38 0.0284367
0.39 0.0274593
0.4 0.0265043
0.41 0.0255711
0.42 0.0246593
0.43 0.0237683
0.44 0.0228978
0.45 0.0220474
0.46 0.0212165
0.47 0.020405
0.48 0.0196124
0.49 0.0188384
0.5 0.0180828
0.51 0.0173451
0.52 0.0166252
0.53 0.0159227
0.54 0.0152374
0.55 0.0145691
0.56 0.0139176
0.57 0.0132826
0.58 0.0126639
0.59 0.0120614
0.6 0.0114748
0.61 0.010904
0.62 0.0103489
0.63 0.00980918
0.64 0.00928482
0.65 0.00877565
0.66 0.00828154
0.67 0.00780238
0.68 0.00733804
0.69 0.00688843
0.7 0.00645344
0.71 0.00603299
0.72 0.005627
0.73 0.00523538
0.74 0.00485806
0.75 0.00449499
0.76 0.00414611
0.77 0.00381136
0.78 0.00349071
0.79 0.00318412
0.8 0.00289154
0.81 0.00261298
0.82 0.00234839
0.83 0.00209778
0.84 0.00186115
0.85 0.00163848
0.86 0.0014298
0.87 0.00123511
0.88 0.00105445
0.89 0.000887848
0.9 0.000735337
0.91 0.000596966
0.92 0.000472792
0.93 0.000362877
0.94 0.000267292
0.95 0.000186121
0.96 0.000119453
0.97 0.00006739
0.98 0.0000300429
0.99 0.00000753471
1. 0
};

\addlegendentry{theoretical}

\addplot[solid, mark=o, mark options={scale=.8}, red]
plot [error bars/.cd, y dir = both, y explicit]
table[y error index=2]{
0.0030 0.1023 0.0079
0.1030 0.0696 0.0060
0.2030 0.0517 0.0050
0.3030 0.0382 0.0041
0.4030 0.0281 0.0036
0.5030 0.0195 0.0029
0.6030 0.0128 0.0024
0.7030 0.0077 0.0018
0.8030 0.0043 0.0013
0.9030 0.0022 0.0009
1.0030 0.0014 0.0006
};
\addlegendentry{$n=1000$}

\addplot[solid, mark=o, mark options={scale=.8}, orange]
plot [error bars/.cd, y dir = both, y explicit]
table[y error index=2]{
0.0010 0.1061 0.0109
0.1010 0.0734 0.0091
0.2010 0.0542 0.0074
0.3010 0.0401 0.0062
0.4010 0.0298 0.0052
0.5010 0.0212 0.0042
0.6010 0.0145 0.0035
0.7010 0.0097 0.0027
0.8010 0.0061 0.0021
0.9010 0.0040 0.0017
1.0010 0.0032 0.0014
};
\addlegendentry{$n=500$}

\addplot[solid, mark=o, mark options={scale=.8}, green]
plot [error bars/.cd, y dir = both, y explicit]
table[y error index=2]{
-0.0010 0.1138 0.0183
0.0990 0.0830 0.0156
0.1990 0.0614 0.0137
0.2990 0.0476 0.0115
0.3990 0.0359 0.0099
0.4990 0.0268 0.0081
0.5990 0.0196 0.0064
0.6990 0.0147 0.0055
0.7990 0.0110 0.0046
0.8990 0.0087 0.0037
0.9990 0.0082 0.0036
};
\addlegendentry{$n=200$}

\addplot[solid, mark=o, mark options={scale=.8}, cyan]
plot [error bars/.cd, y dir = both, y explicit]
table[y error index=2]{
-0.0030 0.1251 0.0252
0.0970 0.0986 0.0248
0.1970 0.0758 0.0220
0.2970 0.0606 0.0186
0.3970 0.0467 0.0170
0.4970 0.0372 0.0146
0.5970 0.0301 0.0125
0.6970 0.0243 0.0107
0.7970 0.0200 0.0092
0.8970 0.0175 0.0088
0.9970 0.0169 0.0083
};
\addlegendentry{$n=100$}

\end{axis}
\end{tikzpicture}

%% file: numerical_fig/unbalancebeta1-hellinger.tex
\begin{tikzpicture}[scale=.5]
\sffamily
\begin{axis}[
title style={align=center,yshift=-.0in},
legend style={font=\large,
 nodes={scale=1, transform shape},
 at={(.7,.9)},
 anchor=north west,
 draw=none},
legend cell align={left},
width = 4.5in, height = 3in,
ylabel near ticks,
ytick={0, .02, .04, .06, .08, .10, .12, .14},
ylabel = {Hellinger-variety},
xlabel near ticks,
every tick label/.append style={font=\normalsize},
xmin=-0.05,xmax=1.05,ymin=0,ymax=.14,
xtick={0, .1, .2, .3, .4, .5, .6, .7, .8, .9, 1},
xlabel={ratio of uninformative participants},
xlabel style = {font=\Large},
ylabel style = {font=\Large},
yticklabel style={
  /pgf/number format/fixed,
  /pgf/number format/precision=5
},
%scaled y ticks=false
%every axis plot/.append style={thick}
]

\addplot[solid, mark=., mark options={scale=.8}, gray, style=thick,]
table[]{
0 0.116269
0.01 0.108875
0.02 0.103724
0.03 0.0993372
0.04 0.0954228
0.05 0.0918481
0.06 0.0885365
0.07 0.0854385
0.08 0.0825199
0.09 0.0797552
0.1 0.0771251
0.11 0.0746144
0.12 0.0722107
0.13 0.069904
0.14 0.0676858
0.15 0.0655491
0.16 0.0634875
0.17 0.061496
0.18 0.0595697
0.19 0.0577047
0.2 0.0558973
0.21 0.0541444
0.22 0.052443
0.23 0.0507905
0.24 0.0491846
0.25 0.0476232
0.26 0.0461042
0.27 0.044626
0.28 0.0431868
0.29 0.0417852
0.3 0.0404197
0.31 0.0390891
0.32 0.0377922
0.33 0.0365279
0.34 0.0352951
0.35 0.0340928
0.36 0.0329203
0.37 0.0317765
0.38 0.0306607
0.39 0.0295722
0.4 0.0285102
0.41 0.0274742
0.42 0.0264634
0.43 0.0254773
0.44 0.0245153
0.45 0.0235769
0.46 0.0226617
0.47 0.021769
0.48 0.0208986
0.49 0.0200498
0.5 0.0192225
0.51 0.0184161
0.52 0.0176303
0.53 0.0168648
0.54 0.0161191
0.55 0.0153931
0.56 0.0146865
0.57 0.0139988
0.58 0.0133299
0.59 0.0126795
0.6 0.0120474
0.61 0.0114332
0.62 0.0108369
0.63 0.0102582
0.64 0.00969681
0.65 0.00915265
0.66 0.00862549
0.67 0.00811517
0.68 0.00762152
0.69 0.00714438
0.7 0.0066836
0.71 0.00623905
0.72 0.00581058
0.73 0.00539808
0.74 0.00500141
0.75 0.00462047
0.76 0.00425516
0.77 0.00390537
0.78 0.003571
0.79 0.00325198
0.8 0.00294821
0.81 0.00265963
0.82 0.00238616
0.83 0.00212774
0.84 0.0018843
0.85 0.0016558
0.86 0.00144219
0.87 0.00124343
0.88 0.00105946
0.89 0.000890277
0.9 0.000735835
0.91 0.000596117
0.92 0.000471104
0.93 0.000360785
0.94 0.000265153
0.95 0.000184204
0.96 0.000117942
0.97 0.0000663757
0.98 0.0000295167
0.99 0.00000738372
1. 0
};

\addlegendentry{theoretical}

\addplot[solid, mark=o, mark options={scale=.8}, red]
plot [error bars/.cd, y dir = both, y explicit]
table[y error index=2]{
0.0030 0.1185 0.0083
0.1030 0.0783 0.0061
0.2030 0.0567 0.0050
0.3030 0.0417 0.0043
0.4030 0.0298 0.0034
0.5030 0.0206 0.0030
0.6030 0.0134 0.0024
0.7030 0.0082 0.0019
0.8030 0.0044 0.0014
0.9030 0.0022 0.0009
1.0030 0.0015 0.0006
};
\addlegendentry{$n=1000$}

\addplot[solid, mark=o, mark options={scale=.8}, orange]
plot [error bars/.cd, y dir = both, y explicit]
table[y error index=2]{
0.0010 0.1223 0.0114
0.1010 0.0816 0.0094
0.2010 0.0593 0.0072
0.3010 0.0433 0.0061
0.4010 0.0318 0.0053
0.5010 0.0225 0.0043
0.6010 0.0151 0.0035
0.7010 0.0097 0.0028
0.8010 0.0063 0.0022
0.9010 0.0038 0.0016
1.0010 0.0033 0.0015
};
\addlegendentry{$n=500$}

\addplot[solid, mark=o, mark options={scale=.8}, green]
plot [error bars/.cd, y dir = both, y explicit]
table[y error index=2]{
-0.0010 0.1314 0.0185
0.0990 0.0925 0.0165
0.1990 0.0673 0.0138
0.2990 0.0501 0.0113
0.3990 0.0375 0.0093
0.4990 0.0279 0.0083
0.5990 0.0203 0.0066
0.6990 0.0149 0.0056
0.7990 0.0111 0.0045
0.8990 0.0090 0.0037
0.9990 0.0080 0.0034
};
\addlegendentry{$n=200$}

\addplot[solid, mark=o, mark options={scale=.8}, cyan]
plot [error bars/.cd, y dir = both, y explicit]
table[y error index=2]{
-0.0030 0.1434 0.0262
0.0970 0.1069 0.0241
0.1970 0.0824 0.0228
0.2970 0.0622 0.0193
0.3970 0.0488 0.0166
0.4970 0.0385 0.0142
0.5970 0.0296 0.0121
0.6970 0.0244 0.0107
0.7970 0.0202 0.0095
0.8970 0.0173 0.0083
0.9970 0.0170 0.0082
};
\addlegendentry{$n=100$}

\end{axis}
\end{tikzpicture}

%% file: numerical_fig/balancebeta2-hellinger.tex
\begin{tikzpicture}[scale=.5]
\sffamily
\begin{axis}[
title style={align=center,yshift=-.0in},
legend style={font=\large,
 nodes={scale=1, transform shape},
 at={(.7,.9)},
 anchor=north west,
 draw=none},
legend cell align={left},
width = 4.5in, height = 3in,
ylabel near ticks,
ytick={0, .02, .04, .06, .08, .10, .12, .14},
ylabel = {Hellinger-variety},
xlabel near ticks,
every tick label/.append style={font=\normalsize},
xmin=-0.05,xmax=1.05,ymin=0,ymax=.14,
xtick={0, .1, .2, .3, .4, .5, .6, .7, .8, .9, 1},
xlabel={ratio of uninformative participants},
xlabel style = {font=\Large},
ylabel style = {font=\Large},
yticklabel style={
  /pgf/number format/fixed,
  /pgf/number format/precision=5
},
%scaled y ticks=false
%every axis plot/.append style={thick}
]

\addplot[solid, mark=., mark options={scale=.8}, gray, style=thick,]
table[]{
0 0.0248359
0.01 0.0236073
0.02 0.0226782
0.03 0.0218615
0.04 0.0211171
0.05 0.0204258
0.06 0.0197764
0.07 0.0191618
0.08 0.0185766
0.09 0.018017
0.1 0.0174801
0.11 0.0169635
0.12 0.0164654
0.13 0.0159841
0.14 0.0155184
0.15 0.015067
0.16 0.014629
0.17 0.0142036
0.18 0.0137901
0.19 0.0133877
0.2 0.0129959
0.21 0.0126142
0.22 0.0122422
0.23 0.0118793
0.24 0.0115252
0.25 0.0111796
0.26 0.0108421
0.27 0.0105125
0.28 0.0101904
0.29 0.00987571
0.3 0.0095681
0.31 0.00926739
0.32 0.00897338
0.33 0.00868588
0.34 0.00840472
0.35 0.00812974
0.36 0.0078608
0.37 0.00759775
0.38 0.00734046
0.39 0.00708881
0.4 0.00684268
0.41 0.00660197
0.42 0.00636656
0.43 0.00613637
0.44 0.0059113
0.45 0.00569126
0.46 0.00547618
0.47 0.00526597
0.48 0.00506056
0.49 0.00485987
0.5 0.00466386
0.51 0.00447244
0.52 0.00428557
0.53 0.00410318
0.54 0.00392523
0.55 0.00375165
0.56 0.00358241
0.57 0.00341745
0.58 0.00325674
0.59 0.00310023
0.6 0.00294788
0.61 0.00279965
0.62 0.00265551
0.63 0.00251543
0.64 0.00237937
0.65 0.0022473
0.66 0.0021192
0.67 0.00199503
0.68 0.00187477
0.69 0.0017584
0.7 0.0016459
0.71 0.00153724
0.72 0.0014324
0.73 0.00133136
0.74 0.00123411
0.75 0.00114062
0.76 0.00105089
0.77 0.000964901
0.78 0.000882634
0.79 0.000804079
0.8 0.000729226
0.81 0.000658064
0.82 0.000590584
0.83 0.000526776
0.84 0.000466633
0.85 0.00041015
0.86 0.000357319
0.87 0.000308136
0.88 0.000262597
0.89 0.000220699
0.9 0.000182439
0.91 0.000147816
0.92 0.000116829
0.93 0.000089478
0.94 0.0000657639
0.95 0.0000456883
0.96 0.0000292537
0.97 0.0000164632
0.98 0.00000732076
0.99 0.0000018312
1. 0
};

\addlegendentry{theoretical}

\addplot[solid, mark=o, mark options={scale=.8}, red]
plot [error bars/.cd, y dir = both, y explicit]
table[y error index=2]{
0.0030 0.0255 0.0041
0.1030 0.0189 0.0032
0.2030 0.0145 0.0028
0.3030 0.0111 0.0024
0.4030 0.0083 0.0020
0.5030 0.0062 0.0017
0.6030 0.0044 0.0014
0.7030 0.0031 0.0011
0.8030 0.0022 0.0010
0.9030 0.0017 0.0007
1.0030 0.0015 0.0007
};
\addlegendentry{$n=1000$}

\addplot[solid, mark=o, mark options={scale=.8}, orange]
plot [error bars/.cd, y dir = both, y explicit]
table[y error index=2]{
0.0010 0.0276 0.0056
0.1010 0.0208 0.0050
0.2010 0.0163 0.0042
0.3010 0.0130 0.0035
0.4010 0.0101 0.0031
0.5010 0.0081 0.0027
0.6010 0.0064 0.0023
0.7010 0.0048 0.0019
0.8010 0.0040 0.0017
0.9010 0.0034 0.0015
1.0010 0.0032 0.0015
};
\addlegendentry{$n=500$}

\addplot[solid, mark=o, mark options={scale=.8}, green]
plot [error bars/.cd, y dir = both, y explicit]
table[y error index=2]{
-0.0010 0.0333 0.0092
0.0990 0.0267 0.0085
0.1990 0.0219 0.0079
0.2990 0.0181 0.0069
0.3990 0.0149 0.0058
0.4990 0.0130 0.0054
0.5990 0.0115 0.0047
0.6990 0.0095 0.0040
0.7990 0.0089 0.0038
0.8990 0.0085 0.0037
0.9990 0.0082 0.0034
};
\addlegendentry{$n=200$}

\addplot[solid, mark=o, mark options={scale=.8}, cyan]
plot [error bars/.cd, y dir = both, y explicit]
table[y error index=2]{
-0.0030 0.0416 0.0143
0.0970 0.0359 0.0135
0.1970 0.0318 0.0127
0.2970 0.0279 0.0116
0.3970 0.0245 0.0109
0.4970 0.0221 0.0104
0.5970 0.0202 0.0092
0.6970 0.0186 0.0085
0.7970 0.0177 0.0081
0.8970 0.0167 0.0077
0.9970 0.0165 0.0077
};
\addlegendentry{$n=100$}

\end{axis}
\end{tikzpicture}

%% file: numerical_fig/unbalancebeta2-hellinger.tex
\begin{tikzpicture}[scale=.5]
\sffamily
\begin{axis}[
title style={align=center,yshift=-.0in},
legend style={font=\large,
 nodes={scale=1, transform shape},
 at={(.7,.9)},
 anchor=north west,
 draw=none},
legend cell align={left},
width = 4.5in, height = 3in,
ylabel near ticks,
ytick={0, .02, .04, .06, .08, .10, .12, .14},
ylabel = {Hellinger-variety},
xlabel near ticks,
every tick label/.append style={font=\normalsize},
xmin=-0.05,xmax=1.05,ymin=0,ymax=.14,
xtick={0, .1, .2, .3, .4, .5, .6, .7, .8, .9, 1},
xlabel={ratio of uninformative participants},
xlabel style = {font=\Large},
ylabel style = {font=\Large},
yticklabel style={
  /pgf/number format/fixed,
  /pgf/number format/precision=5
},
%scaled y ticks=false
%every axis plot/.append style={thick}
]

\addplot[solid, mark=., mark options={scale=.8}, gray, style=thick,]
table[]{
0 0.045614
0.01 0.0434874
0.02 0.0418435
0.03 0.0403861
0.04 0.0390502
0.05 0.0378045
0.06 0.0366309
0.07 0.035517
0.08 0.0344544
0.09 0.0334365
0.1 0.0324585
0.11 0.0315162
0.12 0.0306064
0.13 0.0297265
0.14 0.0288741
0.15 0.0280474
0.16 0.0272447
0.17 0.0264644
0.18 0.0257053
0.19 0.0249663
0.2 0.0242464
0.21 0.0235446
0.22 0.02286
0.23 0.0221921
0.24 0.0215401
0.25 0.0209033
0.26 0.0202813
0.27 0.0196734
0.28 0.0190793
0.29 0.0184985
0.3 0.0179305
0.31 0.017375
0.32 0.0168316
0.33 0.0163001
0.34 0.01578
0.35 0.0152712
0.36 0.0147733
0.37 0.0142861
0.38 0.0138094
0.39 0.0133428
0.4 0.0128863
0.41 0.0124397
0.42 0.0120026
0.43 0.0115751
0.44 0.0111568
0.45 0.0107476
0.46 0.0103475
0.47 0.00995616
0.48 0.00957356
0.49 0.00919956
0.5 0.00883402
0.51 0.00847685
0.52 0.00812792
0.53 0.00778715
0.54 0.00745443
0.55 0.00712967
0.56 0.00681279
0.57 0.0065037
0.58 0.00620234
0.59 0.00590863
0.6 0.00562249
0.61 0.00534387
0.62 0.0050727
0.63 0.00480893
0.64 0.0045525
0.65 0.00430337
0.66 0.00406148
0.67 0.0038268
0.68 0.00359928
0.69 0.00337888
0.7 0.00316557
0.71 0.00295933
0.72 0.00276011
0.73 0.00256789
0.74 0.00238266
0.75 0.00220438
0.76 0.00203304
0.77 0.00186863
0.78 0.00171113
0.79 0.00156053
0.8 0.00141683
0.81 0.00128001
0.82 0.00115007
0.83 0.00102702
0.84 0.000910859
0.85 0.000801583
0.86 0.000699205
0.87 0.000603732
0.88 0.000515176
0.89 0.000433553
0.9 0.000358879
0.91 0.000291174
0.92 0.000230459
0.93 0.000176761
0.94 0.000130106
0.95 0.0000905246
0.96 0.0000580509
0.97 0.0000327207
0.98 0.0000145734
0.99 0.00000365135
1. 0
};

\addlegendentry{theoretical}

\addplot[solid, mark=o, mark options={scale=.8}, red]
plot [error bars/.cd, y dir = both, y explicit]
table[y error index=2]{
0.0030 0.0473 0.0053
0.1030 0.0344 0.0042
0.2030 0.0261 0.0037
0.3030 0.0195 0.0030
0.4030 0.0147 0.0026
0.5030 0.0105 0.0022
0.6030 0.0072 0.0018
0.7030 0.0047 0.0015
0.8030 0.0029 0.0011
0.9030 0.0018 0.0008
1.0030 0.0015 0.0007
};
\addlegendentry{$n=1000$}

\addplot[solid, mark=o, mark options={scale=.8}, orange]
plot [error bars/.cd, y dir = both, y explicit]
table[y error index=2]{
0.0010 0.0497 0.0074
0.1010 0.0367 0.0066
0.2010 0.0276 0.0054
0.3010 0.0216 0.0046
0.4010 0.0167 0.0042
0.5010 0.0124 0.0033
0.6010 0.0089 0.0028
0.7010 0.0065 0.0023
0.8010 0.0045 0.0018
0.9010 0.0036 0.0015
1.0010 0.0031 0.0014
};
\addlegendentry{$n=500$}

\addplot[solid, mark=o, mark options={scale=.8}, green]
plot [error bars/.cd, y dir = both, y explicit]
table[y error index=2]{
-0.0010 0.0560 0.0124
0.0990 0.0437 0.0113
0.1990 0.0346 0.0101
0.2990 0.0274 0.0086
0.3990 0.0222 0.0075
0.4990 0.0180 0.0066
0.5990 0.0140 0.0055
0.6990 0.0116 0.0048
0.7990 0.0098 0.0043
0.8990 0.0087 0.0037
0.9990 0.0082 0.0036
};
\addlegendentry{$n=200$}

\addplot[solid, mark=o, mark options={scale=.8}, cyan]
plot [error bars/.cd, y dir = both, y explicit]
table[y error index=2]{
-0.0030 0.0654 0.0180
0.0970 0.0547 0.0176
0.1970 0.0458 0.0161
0.2970 0.0380 0.0142
0.3970 0.0323 0.0134
0.4970 0.0279 0.0118
0.5970 0.0237 0.0108
0.6970 0.0206 0.0095
0.7970 0.0180 0.0088
0.8970 0.0173 0.0081
0.9970 0.0174 0.0082
};
\addlegendentry{$n=100$}

\end{axis}
\end{tikzpicture}

%% file: case_athletes/watch-sample.tex
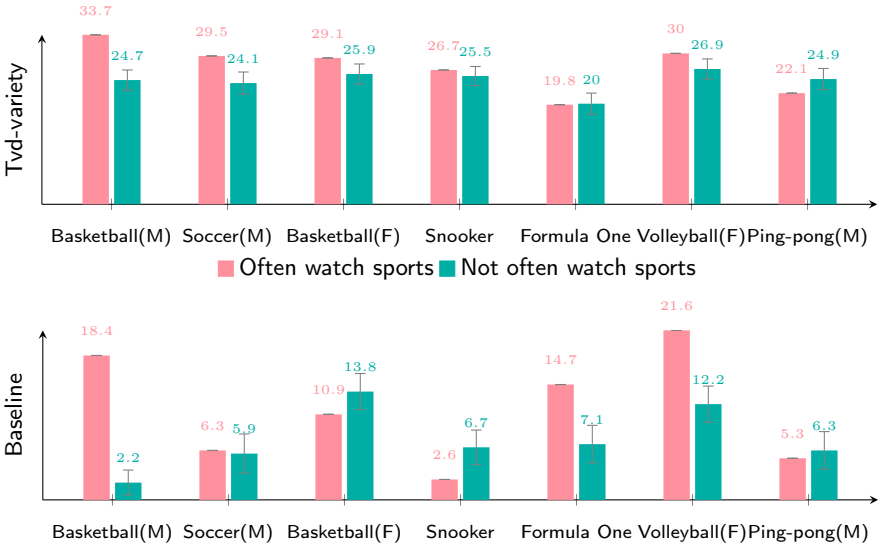
\begin{figure}[!h]
\definecolor{col1}{HTML}{ff919e}
\definecolor{col2}{HTML}{00aea5}
\centering
\begin{tikzpicture}[font=\small]
\sffamily
\begin{axis}[
  ybar,
  bar width=10pt,
  ylabel={Tvd-variety},
  ylabel near ticks,
  width = 5in, height = 1.5in,
  ymin=0,
  ytick=\empty,
  xtick=data,
  tick label style={font=\scriptsize},
  axis x line=bottom,
  axis y line=left,
  enlarge x limits=0.1,
  symbolic x coords={Basketball(M),Soccer(M),Basketball(F),Snooker,Formula One,Volleyball(F),Ping-pong(M)},
  xticklabel style={anchor=base,yshift=-\baselineskip},
  nodes near coords={\pgfmathprintnumber\pgfplotspointmeta},
  every node near coord/.append style={yshift=4pt},
  %x tick label style={font=\tiny,},
  legend image code/.code={
        \draw [#1] (0cm,-0.1cm) rectangle (0.2cm,0.1cm); },
  legend style={
    font=\small,
    nodes={scale=1, transform shape},
    at={(0.5,-0.25)},
    anchor=north,
    draw=none,
    legend columns=-1},
]
  \addplot[font=\tiny,color=col1,bar shift=-6pt,style={fill=col1,draw=none},error bars/.cd, y dir=both, y explicit,error bar style=gray] coordinates {
(Basketball(M),33.7) +- (0,0.0)
(Soccer(M),29.5) +- (0,0.0)
(Basketball(F),29.1) +- (0,0.0)
(Snooker,26.7) +- (0,0.0)
(Formula One,19.8) +- (0,0.0)
(Volleyball(F),30.0) +- (0,0.0)
(Ping-pong(M),22.1) +- (0,0.0)
  };
  \addlegendentry{Often watch sports}
  \addplot[font=\tiny,color=col2,bar shift=6pt,style={fill=col2,draw=none},error bars/.cd, y dir=both, y explicit,error bar style=gray] coordinates {
(Basketball(M),24.7) +- (0,2.0)
(Soccer(M),24.1) +- (0,2.2)
(Basketball(F),25.9) +- (0,2.0)
(Snooker,25.5) +- (0,1.9)
(Formula One,20.0) +- (0,2.1)
(Volleyball(F),26.9) +- (0,2.0)
(Ping-pong(M),24.9) +- (0,2.1)
  };

  \addlegendentry{Not often watch sports}
\end{axis}
\end{tikzpicture}
\begin{tikzpicture}[font=\small]
\sffamily
\begin{axis}[
  ybar,
  bar width=10pt,
  ylabel={{\textcolor{white}{|}Baseline\textcolor{white}{|}}},
  ylabel near ticks,
  width = 5in, height = 1.5in,
  ymin=0,
  ytick=\empty,
  xtick=data,
  tick label style={font=\scriptsize},
  axis x line=bottom,
  axis y line=left,
  enlarge x limits=0.1,
  symbolic x coords={Basketball(M),Soccer(M),Basketball(F),Snooker,Formula One,Volleyball(F),Ping-pong(M)},
  xticklabel style={anchor=base,yshift=-\baselineskip},
  nodes near coords={\pgfmathprintnumber\pgfplotspointmeta},
  every node near coord/.append style={yshift=4pt},
  %x tick label style={font=\tiny,},
  legend image code/.code={
        \draw [#1] (0cm,-0.1cm) rectangle (0.2cm,0.1cm); },
  legend style={
    font=\small,
    nodes={scale=1, transform shape},
    at={(0.5,-0.25)},
    anchor=north,
    draw=none,
    legend columns=-1},
]
  \addplot[font=\tiny,color=col1,bar shift=-6pt,style={fill=col1,draw=none},error bars/.cd, y dir=both, y explicit,error bar style=gray] coordinates {
(Basketball(M),18.4) +- (0,0.0)
(Soccer(M),6.3) +- (0,0.0)
(Basketball(F),10.9) +- (0,0.0)
(Snooker,2.6) +- (0,0.0)
(Formula One,14.7) +- (0,0.0)
(Volleyball(F),21.6) +- (0,0.0)
(Ping-pong(M),5.3) +- (0,0.0)
  };
  
  \addplot[font=\tiny,color=col2,bar shift=6pt,style={fill=col2,draw=none},error bars/.cd, y dir=both, y explicit,error bar style=gray] coordinates {
(Basketball(M),2.2) +- (0,1.6)
(Soccer(M),5.9) +- (0,2.5)
(Basketball(F),13.8) +- (0,2.3)
(Snooker,6.7) +- (0,2.2)
(Formula One,7.1) +- (0,2.4)
(Volleyball(F),12.2) +- (0,2.3)
(Ping-pong(M),6.3) +- (0,2.4)
  };
  
\end{axis}
\end{tikzpicture}
\captionof{figure}{\small {\bf Often watching sports  {\sl v.s.} not}}
\label{ath:watch-sample}
\end{figure}

%% file: case_talkshow/native-sample.tex
\begin{figure}[!h]
\definecolor{col1}{HTML}{ff919e}
\definecolor{col2}{HTML}{00aea5}
\centering
\begin{tikzpicture}[font=\small]
\sffamily
\begin{axis}[
  ybar,
  bar width=10pt,
  ylabel={Tvd-variety},
  ylabel near ticks,
  width = 5in, height = 1.5in,
  ymin=0,
  ytick=\empty,
  xtick=data,
  tick label style={font=\scriptsize},
  axis x line=bottom,
  axis y line=left,
  enlarge x limits=0.1,
  symbolic x coords={Native1,Native2,Native3,Native4},
  xticklabel style={anchor=base,yshift=-\baselineskip},
  nodes near coords={\pgfmathprintnumber\pgfplotspointmeta},
  every node near coord/.append style={yshift=4pt},
  %x tick label style={font=\tiny,},
  legend image code/.code={
        \draw [#1] (0cm,-0.1cm) rectangle (0.2cm,0.1cm); },
  legend style={
    font=\small,
    nodes={scale=1, transform shape},
    at={(0.5,-0.25)},
    anchor=north,
    draw=none,
    legend columns=-1},
]
  \addplot[font=\tiny,color=col1,bar shift=-6pt,style={fill=col1,draw=none},error bars/.cd, y dir=both, y explicit,error bar style=gray] coordinates {
(Native1,38.8) +- (0,2.2)
(Native2,30.9) +- (0,2.5)
(Native3,33.7) +- (0,2.6)
(Native4,33.9) +- (0,2.5)
  };
  \addlegendentry{\footnotesize Familiar with native stand-up comedy}
  \addplot[font=\tiny,color=col2,bar shift=6pt,style={fill=col2,draw=none},error bars/.cd, y dir=both, y explicit,error bar style=gray] coordinates {
(Native1,35.7) +- (0,0.0)
(Native2,26.4) +- (0,0.0)
(Native3,30.0) +- (0,0.0)
(Native4,30.0) +- (0,0.0)
  };

  \addlegendentry{\footnotesize Unfamiliar with native stand-up comedy}
\end{axis}
\end{tikzpicture}
\begin{tikzpicture}[font=\small]
\sffamily
\begin{axis}[
  ybar,
  bar width=10pt,
  ylabel={{\textcolor{white}{|}Baseline\textcolor{white}{|}}},
  ylabel near ticks,
  width = 5in, height = 1.5in,
  ymin=0,
  ytick=\empty,
  xtick=data,
  tick label style={font=\scriptsize},
  axis x line=bottom,
  axis y line=left,
  enlarge x limits=0.1,
  symbolic x coords={Native1,Native2,Native3,Native4},
  xticklabel style={anchor=base,yshift=-\baselineskip},
  nodes near coords={\pgfmathprintnumber\pgfplotspointmeta},
  every node near coord/.append style={yshift=4pt},
  %x tick label style={font=\tiny,},
  legend image code/.code={
        \draw [#1] (0cm,-0.1cm) rectangle (0.2cm,0.1cm); },
  legend style={
    font=\small,
    nodes={scale=1, transform shape},
    at={(0.5,-0.25)},
    anchor=north,
    draw=none,
    legend columns=-1},
]
  \addplot[font=\tiny,color=col1,bar shift=-6pt,style={fill=col1,draw=none},error bars/.cd, y dir=both, y explicit,error bar style=gray] coordinates {
(Native1,12.9) +- (0,3.5)
(Native2,15.2) +- (0,3.4)
(Native3,10.5) +- (0,3.4)
(Native4,16.9) +- (0,3.5)
  };
  
  \addplot[font=\tiny,color=col2,bar shift=6pt,style={fill=col2,draw=none},error bars/.cd, y dir=both, y explicit,error bar style=gray] coordinates {
(Native1,15.0) +- (0,0.0)
(Native2,17.1) +- (0,0.0)
(Native3,12.9) +- (0,0.0)
(Native4,15.7) +- (0,0.0)
  };
  
\end{axis}
\end{tikzpicture}
\captionof{figure}{\small {\bf Familiar with native stand-up comedy {\sl v.s.} unfamiliar}}
\label{ath:sex_sample1}
\end{figure}
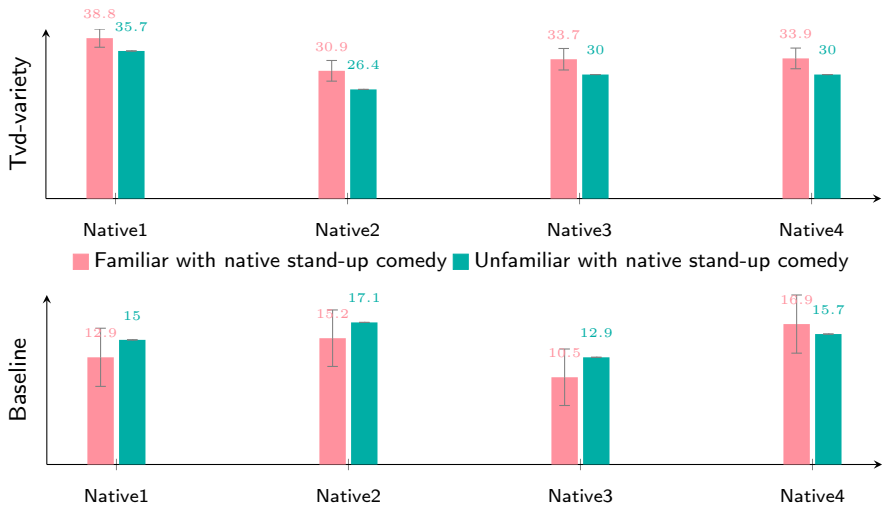

%% file: case_talkshow/foreign-sample.tex
\begin{figure}[!h]
\definecolor{col1}{HTML}{ff919e}
\definecolor{col2}{HTML}{00aea5}
\centering
\begin{tikzpicture}[font=\small]
\sffamily
\begin{axis}[
  ybar,
  bar width=10pt,
  ylabel={Tvd-variety},
  ylabel near ticks,
  width = 5in, height = 1.5in,
  ymin=0,
  ytick=\empty,
  xtick=data,
  tick label style={font=\scriptsize},
  axis x line=bottom,
  axis y line=left,
  enlarge x limits=0.1,
  symbolic x coords={Foreign1,Foreign2,Foreign3,Foreign4},
  xticklabel style={anchor=base,yshift=-\baselineskip},
  nodes near coords={\pgfmathprintnumber\pgfplotspointmeta},
  every node near coord/.append style={yshift=4pt},
  %x tick label style={font=\tiny,},
  legend image code/.code={
        \draw [#1] (0cm,-0.1cm) rectangle (0.2cm,0.1cm); },
  legend style={
    font=\small,
    nodes={scale=1, transform shape},
    at={(0.5,-0.25)},
    anchor=north,
    draw=none,
    legend columns=-1},
]
  \addplot[font=\tiny,color=col1,bar shift=-6pt,style={fill=col1,draw=none},error bars/.cd, y dir=both, y explicit,error bar style=gray] coordinates {
(Foreign1,31.6) +- (0,0.0)
(Foreign2,22.6) +- (0,0.0)
(Foreign3,30.0) +- (0,0.0)
(Foreign4,31.1) +- (0,0.0)
  };
  \addlegendentry{\footnotesize Familiar with foreign stand-up comedy}
  \addplot[font=\tiny,color=col2,bar shift=6pt,style={fill=col2,draw=none},error bars/.cd, y dir=both, y explicit,error bar style=gray] coordinates {
(Foreign1,28.1) +- (0,2.2)
(Foreign2,23.9) +- (0,2.4)
(Foreign3,28.5) +- (0,2.3)
(Foreign4,27.3) +- (0,2.3)
  };

  \addlegendentry{\footnotesize Unfamiliar with foreign stand-up comedy}
\end{axis}
\end{tikzpicture}
\begin{tikzpicture}[font=\small]
\sffamily
\begin{axis}[
  ybar,
  bar width=10pt,
  ylabel={{\textcolor{white}{|}Baseline\textcolor{white}{|}}},
  ylabel near ticks,
  width = 5in, height = 1.5in,
  ymin=0,
  ytick=\empty,
  xtick=data,
  tick label style={font=\scriptsize},
  axis x line=bottom,
  axis y line=left,
  enlarge x limits=0.1,
  symbolic x coords={Foreign1,Foreign2,Foreign3,Foreign4},
  xticklabel style={anchor=base,yshift=-\baselineskip},
  nodes near coords={\pgfmathprintnumber\pgfplotspointmeta},
  every node near coord/.append style={yshift=4pt},
  %x tick label style={font=\tiny,},
  legend image code/.code={
        \draw [#1] (0cm,-0.1cm) rectangle (0.2cm,0.1cm); },
  legend style={
    font=\small,
    nodes={scale=1, transform shape},
    at={(0.5,-0.25)},
    anchor=north,
    draw=none,
    legend columns=-1},
]
  \addplot[font=\tiny,color=col1,bar shift=-6pt,style={fill=col1,draw=none},error bars/.cd, y dir=both, y explicit,error bar style=gray] coordinates {
(Foreign1,15.8) +- (0,0.0)
(Foreign2,2.6) +- (0,0.0)
(Foreign3,16.8) +- (0,0.0)
(Foreign4,3.7) +- (0,0.0)
  };
  
  \addplot[font=\tiny,color=col2,bar shift=6pt,style={fill=col2,draw=none},error bars/.cd, y dir=both, y explicit,error bar style=gray] coordinates {
(Foreign1,7.6) +- (0,2.9)
(Foreign2,3.8) +- (0,2.4)
(Foreign3,14.1) +- (0,2.7)
(Foreign4,2.5) +- (0,1.9)
  };
  
\end{axis}
\end{tikzpicture}
\captionof{figure}{\small {\bf Familiar with foreign stand-up comedy {\sl v.s.} unfamiliar}}
\label{ath:sex_sample2}
\end{figure}
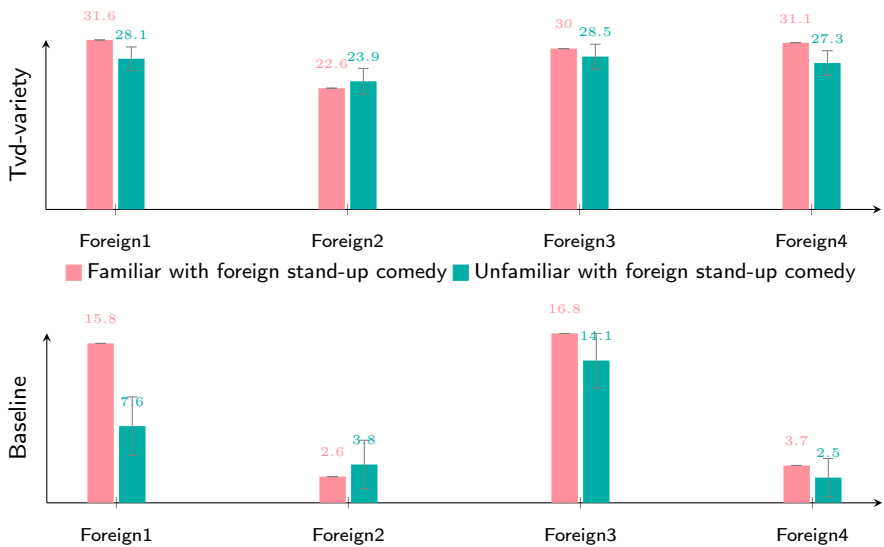

%% file: case_athletes/questionnaire.tex
\begin{enumerate}[itemsep=1pt]
    \item What is your gender?
    
    \begin{enumerate*}[itemjoin=\qquad]
        \item Female
        \item Male
    \end{enumerate*}
    
    \item Do you often watch sports?
    
    \begin{enumerate*}[itemjoin=\qquad]
        \item I often watch sports
        \item I do not often watch sports
    \end{enumerate*}
    
    \item Which of the following two basketball players do you prefer?
    
    \begin{enumerate*}[itemjoin=\qquad]
        \item Zhenlin Zhang
        \item Songwei Zhu
    \end{enumerate*}
    
    \item What percentage of people do you think prefer Zhenlin Zhang?
    
    \begin{enumerate*}[itemjoin=\ ]
        \item 0\item 10\item 20\item 30\item 40\item 50
        \item 60\item 70\item 80\item 90\item 100
    \end{enumerate*}
    
    \item Which of the following two soccer players do you prefer?
    
    \begin{enumerate*}[itemjoin=\qquad]
        \item Andrés Iniesta
        \item Luka Modrić
    \end{enumerate*}
    
    \item What percentage of people do you think prefer Andrés Iniesta?
    
    \begin{enumerate*}[itemjoin=\ ]
        \item 0\item 10\item 20\item 30\item 40\item 50
        \item 60\item 70\item 80\item 90\item 100
    \end{enumerate*}
    
    \item There are 8 red balls and 12 blue balls with the same shape in the box. One is randomly selected. What percentage do you think is the probability of a blue ball?
    
    \begin{enumerate*}[itemjoin=\ ]
        \item 0\item 10\item 20\item 30\item 40\item 50
        \item 60\item 70\item 80\item 90\item 100
    \end{enumerate*}
    
    \item Which of the following two basketball players do you prefer?
    
    \begin{enumerate*}[itemjoin=\qquad]
        \item Nan Chen
        \item Lijie Miao
    \end{enumerate*}
    
    \item What percentage of people do you think prefer Nan Chen?
    
    \begin{enumerate*}[itemjoin=\ ]
        \item 0\item 10\item 20\item 30\item 40\item 50
        \item 60\item 70\item 80\item 90\item 100
    \end{enumerate*}
    
    \item Which of the following two snooker players do you prefer?
    
    \begin{enumerate*}[itemjoin=\qquad]
        \item Judd Trump
        \item John Higgins
    \end{enumerate*}
    
    \item What percentage of people do you think prefer Judd Trump?
    
    \begin{enumerate*}[itemjoin=\ ]
        \item 0\item 10\item 20\item 30\item 40\item 50
        \item 60\item 70\item 80\item 90\item 100
    \end{enumerate*}
    
    \item Which of the following two Formula One players do you prefer?
    
    \begin{enumerate*}[itemjoin=\qquad]
        \item Sebastian Vettel
        \item Lewis Hamilton
    \end{enumerate*}
    
    \item What percentage of people do you think prefer Sebastian Vettel?
    
    \begin{enumerate*}[itemjoin=\ ]
        \item 0\item 10\item 20\item 30\item 40\item 50
        \item 60\item 70\item 80\item 90\item 100
    \end{enumerate*}
    
    \item Which of the following two volleyball players do you prefer?
    
    \begin{enumerate*}[itemjoin=\qquad]
        \item Ruirui Zhao
        \item Yimei Wang
    \end{enumerate*}
    
    \item What percentage of people do you think prefer Ruirui Zhao?
    
    \begin{enumerate*}[itemjoin=\ ]
        \item 0\item 10\item 20\item 30\item 40\item 50
        \item 60\item 70\item 80\item 90\item 100
    \end{enumerate*}
    
    \item Which of the following two ping-pong players do you prefer?
    
    \begin{enumerate*}[itemjoin=\qquad]
        \item Jingkun Liang
        \item Chuqin Wang
    \end{enumerate*}
    
    \item What percentage of people do you think prefer Jingkun Liang?
    
    \begin{enumerate*}[itemjoin=\ ]
        \item 0\item 10\item 20\item 30\item 40\item 50
        \item 60\item 70\item 80\item 90\item 100
    \end{enumerate*}
\end{enumerate}

%% file: case_talkshow/questionnaire.tex
\begin{enumerate}[itemsep=1pt]
    \item What is your gender?
    
    \begin{enumerate*}[itemjoin=\qquad]
        \item Female
        \item Male
    \end{enumerate*}
    
    \item How often do you watch native stand-up comedies?
    
    \begin{enumerate*}[itemjoin=\qquad]
        \item often
        \item sometimes
        \item occasionally
        \item almost never
    \end{enumerate*}
    
    \item How often do you watch foreign stand-up comedies?
    
    \begin{enumerate*}[itemjoin=\qquad]
        \item often
        \item sometimes
        \item occasionally
        \item almost never
    \end{enumerate*}
    
    \item Which of the following two stand-up comedians do you prefer?
    
    \begin{enumerate*}[itemjoin=\qquad]
        \item Lan Hu
        \item Jianguo Wang
    \end{enumerate*}
    
    \item What percentage of people do you think prefer Lan Hu?
    
    \begin{enumerate*}[itemjoin=\ ]
        \item 0\item 10\item 20\item 30\item 40\item 50
        \item 60\item 70\item 80\item 90\item 100
    \end{enumerate*}
    
    \item Which of the following two stand-up comedians do you prefer?
    
    \begin{enumerate*}[itemjoin=\qquad]
        \item Whitney Cummings
        \item Ali Wong
    \end{enumerate*}
    
    \item What percentage of people do you think prefer Whitney Cummings?
    
    \begin{enumerate*}[itemjoin=\ ]
        \item 0\item 10\item 20\item 30\item 40\item 50
        \item 60\item 70\item 80\item 90\item 100
    \end{enumerate*}
    
    \item Which of the following two stand-up comedians do you prefer?
    
    \begin{enumerate*}[itemjoin=\qquad]
        \item Guangzhi He
        \item Mengen Yang
    \end{enumerate*}
    
    \item What percentage of people do you think prefer Guangzhi He?
    
    \begin{enumerate*}[itemjoin=\ ]
        \item 0\item 10\item 20\item 30\item 40\item 50
        \item 60\item 70\item 80\item 90\item 100
    \end{enumerate*}
    
    \item Which of the following two stand-up comedians do you prefer?
    
    \begin{enumerate*}[itemjoin=\qquad]
        \item Russell Peters
        \item Michael McIntyre
    \end{enumerate*}
    
    \item What percentage of people do you think prefer Russell Peters?
    
    \begin{enumerate*}[itemjoin=\ ]
        \item 0\item 10\item 20\item 30\item 40\item 50
        \item 60\item 70\item 80\item 90\item 100
    \end{enumerate*}
    
    \item Which of the following two stand-up comedians do you prefer?
    
    \begin{enumerate*}[itemjoin=\qquad]
        \item Li Yang
        \item Siwen Wang
    \end{enumerate*}
    
    \item What percentage of people do you think prefer Li Yang?
    
    \begin{enumerate*}[itemjoin=\ ]
        \item 0\item 10\item 20\item 30\item 40\item 50
        \item 60\item 70\item 80\item 90\item 100
    \end{enumerate*}
    
    \item Which of the following two stand-up comedians do you prefer?
    
    \begin{enumerate*}[itemjoin=\qquad]
        \item KT Tatara
        \item Yumi Nagashima
    \end{enumerate*}
    
    \item What percentage of people do you think prefer KT Tatara?
    
    \begin{enumerate*}[itemjoin=\ ]
        \item 0\item 10\item 20\item 30\item 40\item 50
        \item 60\item 70\item 80\item 90\item 100
    \end{enumerate*}
    
    \item There are 1 red balls and 4 blue balls with the same shape in the box. One is randomly selected. What percentage do you think is the probability of a red ball?
    
    \begin{enumerate*}[itemjoin=\ ]
        \item 0\item 10\item 20\item 30\item 40\item 50
        \item 60\item 70\item 80\item 90\item 100
    \end{enumerate*}
    
    \item Which of the following two stand-up comedians do you prefer?
    
    \begin{enumerate*}[itemjoin=\qquad]
        \item Bo Pang
        \item Qimo Zhou
    \end{enumerate*}
    
    \item What percentage of people do you think prefer Bo Pang?
    
    \begin{enumerate*}[itemjoin=\ ]
        \item 0\item 10\item 20\item 30\item 40\item 50
        \item 60\item 70\item 80\item 90\item 100
    \end{enumerate*}
    
    \item Which of the following two stand-up comedians do you prefer?
    
    \begin{enumerate*}[itemjoin=\qquad]
        \item Ronny Chieng
        \item Jimmy OYang
    \end{enumerate*}
    
    \item What percentage of people do you think prefer Ronny Chieng?
    
    \begin{enumerate*}[itemjoin=\ ]
        \item 0\item 10\item 20\item 30\item 40\item 50
        \item 60\item 70\item 80\item 90\item 100
    \end{enumerate*}
    
\end{enumerate}

%% file: main.bbl
\begin{thebibliography}{14}
\providecommand{\natexlab}[1]{#1}
\providecommand{\url}[1]{\texttt{#1}}
\expandafter\ifx\csname urlstyle\endcsname\relax
  \providecommand{\doi}[1]{doi: #1}\else
  \providecommand{\doi}{doi: \begingroup \urlstyle{rm}\Url}\fi

\bibitem[Csisz{\'a}r et~al.(2004)Csisz{\'a}r, Shields,
  et~al.]{csiszar2004information}
Imre Csisz{\'a}r, Paul~C Shields, et~al.
\newblock Information theory and statistics: A tutorial.
\newblock \emph{Foundations and Trends{\textregistered} in Communications and
  Information Theory}, 1\penalty0 (4):\penalty0 417--528, 2004.

\bibitem[Galesic et~al.(2018)Galesic, de~Bruin, Dumas, Kapteyn, Darling, and
  Meijer]{galesic2018asking}
Mirta Galesic, W~Bruine de~Bruin, Marion Dumas, A~Kapteyn, JE~Darling, and
  E~Meijer.
\newblock Asking about social circles improves election predictions.
\newblock \emph{Nature Human Behaviour}, 2\penalty0 (3):\penalty0 187--193,
  2018.

\bibitem[Helzer and Dunning(2012)]{helzer2012and}
Erik~G Helzer and David Dunning.
\newblock Why and when peer prediction is superior to self-prediction: The
  weight given to future aspiration versus past achievement.
\newblock \emph{Journal of personality and social psychology}, 103\penalty0
  (1):\penalty0 38, 2012.

\bibitem[Johnson et~al.(1995)Johnson, Kotz, and
  Balakrishnan]{johnson1995continuous}
Norman~L Johnson, Samuel Kotz, and Narayanaswamy Balakrishnan.
\newblock \emph{Continuous univariate distributions, volume 2}, volume 289.
\newblock John wiley \& sons, 1995.

\bibitem[Kong and
  Schoenebeck(2018{\natexlab{a}})]{DBLP:conf/innovations/KongS18}
Yuqing Kong and Grant Schoenebeck.
\newblock Equilibrium selection in information elicitation without verification
  via information monotonicity.
\newblock In \emph{9th Innovations in Theoretical Computer Science Conference,
  {ITCS} 2018, January 11-14, 2018, Cambridge, MA, {USA}}
  \citet{DBLP:conf/innovations/KongS18}, pages 13:1--13:20.
\newblock \doi{10.4230/LIPIcs.ITCS.2018.13}.
\newblock URL \url{https://doi.org/10.4230/LIPIcs.ITCS.2018.13}.

\bibitem[Kong and Schoenebeck(2018{\natexlab{b}})]{kong2018water}
Yuqing Kong and Grant Schoenebeck.
\newblock Water from two rocks: Maximizing the mutual information.
\newblock In \emph{Proceedings of the 2018 ACM Conference on Economics and
  Computation}, pages 177--194, 2018{\natexlab{b}}.

\bibitem[Kong and Schoenebeck(2019)]{kong2019information}
Yuqing Kong and Grant Schoenebeck.
\newblock An information theoretic framework for designing information
  elicitation mechanisms that reward truth-telling.
\newblock \emph{ACM Transactions on Economics and Computation (TEAC)},
  7\penalty0 (1):\penalty0 1--33, 2019.

\bibitem[Prelec(2004)]{prelec2004bayesian}
Dra{\v{z}}en Prelec.
\newblock A bayesian truth serum for subjective data.
\newblock \emph{science}, 306\penalty0 (5695):\penalty0 462--466, 2004.

\bibitem[Prelec et~al.(2017)Prelec, Seung, and McCoy]{prelec2017solution}
Dra{\v{z}}en Prelec, H~Sebastian Seung, and John McCoy.
\newblock A solution to the single-question crowd wisdom problem.
\newblock \emph{Nature}, 541\penalty0 (7638):\penalty0 532--535, 2017.

\bibitem[Radanovic and Faltings(2014)]{radanovic2014incentives}
Goran Radanovic and Boi Faltings.
\newblock Incentives for truthful information elicitation of continuous
  signals.
\newblock In \emph{Proceedings of the AAAI Conference on Artificial
  Intelligence}, volume~28, 2014.

\bibitem[Radas and Prelec(2019)]{radas2019whose}
Sonja Radas and Drazen Prelec.
\newblock Whose data can we trust: How meta-predictions can be used to uncover
  credible respondents in survey data.
\newblock \emph{PloS one}, 14\penalty0 (12):\penalty0 e0225432, 2019.

\bibitem[Rothschild and Wolfers(2011)]{rothschild2011forecasting}
David Rothschild and Justin Wolfers.
\newblock Forecasting elections: Voter intentions versus expectations.
\newblock \emph{Available at SSRN 1884644}, 2011.

\bibitem[Weaver and Prelec(2013)]{weaver2013creating}
Ray Weaver and Drazen Prelec.
\newblock Creating truth-telling incentives with the bayesian truth serum.
\newblock \emph{Journal of Marketing Research}, 50\penalty0 (3):\penalty0
  289--302, 2013.

\bibitem[Witkowski and Parkes(2012)]{witkowski2012robust}
Jens Witkowski and David Parkes.
\newblock A robust bayesian truth serum for small populations.
\newblock In \emph{Proceedings of the AAAI Conference on Artificial
  Intelligence}, volume~26, 2012.

\end{thebibliography}
